\documentclass[letterpaper,11pt]{article}

\AtBeginDocument{%
  }
    
\newif\ifdraft
\drafttrue
\newif\iffull
\fullfalse
\usepackage[utf8]{inputenc}
\usepackage{microtype}

\iftrue 
\usepackage[letterpaper,margin=1in]{geometry}
\usepackage[numbers,sort&compress]{natbib}
\usepackage{xcolor,xspace}
\usepackage{amsmath}
\usepackage{amssymb}
\usepackage{amsthm}
\usepackage{paralist}
\usepackage{booktabs}
\usepackage[unicode]{hyperref}
\usepackage[small]{caption}
\fi
\usepackage{enumerate}
\usepackage{textgreek}
\usepackage{graphicx}
\usepackage{framed}
\usepackage{tcolorbox}
\usepackage{tikz}
\usetikzlibrary{positioning,shapes.geometric,patterns}

\usepackage[capitalize, nameinlink]{cleveref}

\usepackage{doi}
\usepackage{thmtools,thm-restate}
\usepackage{xspace}
\usepackage{nicefrac}
\usepackage[noend]{algorithmic}
\usepackage{algorithm}
\usepackage{dsfont}

\usepackage{mathtools}

\theoremstyle{plain}
\newtheorem{theorem}{Theorem}[section]
\newtheorem*{theorem*}{Theorem}
\newtheorem{lemma}[theorem]{Lemma}
\newtheorem*{lemma*}{Lemma}

\newtheorem{corollary}[theorem]{Corollary}
\newtheorem{claim}[theorem]{Claim}
\newtheorem*{claim*}{Claim}

\newtheorem*{conjecture*}{Conjecture}
\newtheorem{observation}[theorem]{Observation}

\theoremstyle{definition}
\newtheorem{definition}[theorem]{Definition}

\theoremstyle{remark}
\newtheorem{remark}[theorem]{Remark}
\newtheorem*{remark*}{Remark}

\newcommand{\trycolor}{\textsc{TryColor}\xspace}
\newcommand{\slackgeneration}{\textsc{SlackGeneration}\xspace}

\newcommand{\CONGEST}{\ensuremath{\mathsf{CONGEST}}\xspace}
\newcommand{\LOCAL}{\ensuremath{\mathsf{LOCAL}}\xspace}

\newcommand{\eps}{\varepsilon}
\newcommand{\poly}{\operatorname{\text{{\rm poly}}}}

\DeclareMathOperator{\E}{\mathbb{E}}

\newcommand{\dist}{\operatorname{dist}}
\newcommand{\One}{\mathds{1}}

\newcommand{\lovasz}{Lov\'{a}sz\xspace}

\newcommand{\vbl}{\textsf{vbl}}

\newcommand{\bandwidth}{\ensuremath{\mathsf{bandwidth}}}
\newcommand{\post}{\ensuremath{\mathsf{post}}}
\newcommand{\pre}{\ensuremath{\mathsf{pre}}}
\newcommand{\passive}{\ensuremath{\mathsf{passive}}}

\newcommand{\Exp}{\mathbb{E}}

\renewcommand{\phi}{\varphi}

\newcommand{\cB}{\mathcal{B}}
\newcommand{\cC}{\mathcal{C}}

\newcommand{\cE}{\mathcal{E}}

\newcommand{\cH}{\mathcal{H}}

\newcommand{\cL}{\mathcal{L}}

\newcommand{\cV}{\mathcal{V}}

\newcommand{\myparagraph}[1]{
\smallskip

\noindent\textbf{#1}}

\newcommand{\risk}{risk\xspace}
\newcommand{\Risk}{Risk\xspace}

\newcommand{\black}{\ensuremath{\mathsf{black}}\xspace}
\newcommand{\white}{\ensuremath{\mathsf{white}}\xspace}
\newcommand{\assoc}{\ensuremath{\mathsf{assoc}}\xspace}
\newcommand{\danger}{\ensuremath{\mathsf{danger}}\xspace}

\newenvironment{mycover}
{\list{}{\listparindent 0pt
        \itemindent    \listparindent
        \leftmargin    1cm
        \rightmargin   1cm
        \parsep        0pt}%
    \raggedright
    \item\relax}
{\endlist}

\newcommand{\myemail}[1]{\,$\cdot$\, {\small #1}}
\newcommand{\myaff}[1]{\,$\cdot$\, {\small #1}\par\smallskip}

\hypersetup{
    colorlinks=true,
    linkcolor=black,
    citecolor=black,
    filecolor=black,
    urlcolor=[HTML]{0088CC},
    pdftitle={Distributed Lovasz Local Lemma under Bandwidth Limitations},
    pdfauthor={Anonymous}
}

\title{Distributed Lovász Local Lemma under Bandwidth Limitations}

\begin{document}

\begin{mycover}
	{\huge\bfseries\boldmath Distributed \lovasz Local Lemma under Bandwidth Limitations \par}
	\bigskip
	\bigskip
	\bigskip
	
	\textbf{Magn\'us M. Halld\'orsson}
	\myemail{mmh@ru.is}
	\myaff{Reykjavik University, Iceland}
	
	\textbf{Yannic Maus\footnote{Supported by the Austrian Science Fund (FWF), Grant P36280-N.}}
	\myemail{yannic.maus@ist.tugraz.at}
	\myaff{TU Graz, Austria}
	
	\textbf{Saku Peltonen}
	\myemail{saku.peltonen@gmail.com}
	\myaff{Aalto University, Finland}
\end{mycover}

\thispagestyle{empty}
\begin{abstract}
The constructive \lovasz Local Lemma has become a central tool for designing efficient distributed algorithms. While it has been extensively studied in the classic LOCAL model that uses unlimited bandwidth, much less is known in the bandwidth-restricted CONGEST model. 

In this paper, we present bandwidth- and time-efficient algorithms for various subclasses of LLL problems, including a large class of subgraph sampling problems that are naturally formulated as LLLs. 
Lastly, we use our LLLs to design efficient CONGEST algorithms for coloring sparse and triangle-free graphs with few colors. These coloring algorithms are exponentially faster than previous LOCAL model algorithms. 
\end{abstract}
\clearpage
\thispagestyle{empty}
\tableofcontents
\clearpage
\setcounter{page}{1}

\section{Introduction}
The \lovasz Local Lemma (LLL) is a powerful probabilistic tool that provides conditions under which many mildly locally dependent ``bad events'' defined over some random variables can simultaneously be avoided. In its computational version, one aims at also computing an assignment of the random variables avoiding all these events. 
Due to its local nature, it has become an important cornerstone for distributed computation, e.g.,  \cite{LLL_lowerbound,CP19,Brandt19,CPS17,MU21,FGLLL17,RG20,GHK18,Davies23}. For example, it plays a key role in the complexity theory of local graph problems, as it is complete for sublogarithmic computation \cite{CP19,MU21}, in the sense that any local graph problem (formally, any locally checkable labeling problem \cite{naor95}) that admits a $o(\log n)$-round 
algorithm can be solved in the same asymptotic runtime as LLL. Other examples are its usage as a tool for splitting graphs into smaller subgraphs while satisfying certain constraints for divide-and-conquer approaches \cite{Davies23}, or to solve concrete problems such as computing edge colorings with few colors \cite{CHLPU20,HMN22,Davies23}. It has also been imperative for developing the award-winning round elimination lower bound technique \cite{LLL_lowerbound,Brandt19,BBHORS21}. 

Classically, research on local distributed graph problems has a strong focus on problems that decompose nicely and are trivially solvable sequentially by greedy algorithms, such as finding maximal independent sets. For many other problems, especially those that go beyond the greedy regime such as many partitioning problems or coloring problems with few colors, easy randomized algorithms hold for high-degree graphs, but low-to-medium ranged degrees prove to be challenging. Combinatorially, LLL is the go-to method for such problems, and efficient LLL algorithms allow for many non-trivial results to immediately carry over \cite{CPS17,HMN22,MU21}.

The breakthrough Moser-Tardos algorithm and a more efficient distributed implementation by Chang, Pettie, and Su yield logarithmic-time distributed algorithms \cite{MoserTardos10,CPS17}. After the publication of a seminal lower bound of $\Omega(\log\log n)$-rounds \cite{LLL_lowerbound}, the quest for understanding under which circumstances sublogarithmic time, optimally $\poly\log\log n$ time, LLL algorithms exist has begun \cite{FGLLL17,CP19,GHK18,Davies23}. 
This algorithmic success has almost exclusively been in the bandwidth-unrestricted setting, with progress on bandwidth-restricted algorithms being significantly more limited. 

\begin{tcolorbox}
In this paper, we provide bandwidth- and time-efficient distributed algorithms for important subclasses of LLL problems and exemplify their usefulness via several applications mainly in the domains of subgraph sampling and coloring sparse graphs with few colors. 
\end{tcolorbox}
To explain the challenges that occur in developing bandwidth-efficient LLL algorithms we begin with the necessary background.

\paragraph{Distributed \lovasz Local Lemma (LLL)}
 An instance $\cL=(\cV,\cB)$ of the \emph{distributed \lovasz Local Lemma (LLL)} is given by  
a set of independent random variables $\cV$ and a family of "bad" events $\cB$ over these variables. The
\lovasz Local Lemma \cite{LLL73} states (in its basic form) that there exists an assignment to the variables that avoids all bad events as long as an appropriate relationship holds between the probability of each bad event and the number of events that any given event depends on. This relationship is called the LLL criterion, with many algorithms assuming stronger criteria than the one required for the existence of a solution. The dependency graph $\cH$ of an LLL has a node for each bad event in $\cB$ with two \emph{bad event nodes}  adjacent if they share a variable. 

\myparagraph{\LOCAL and \CONGEST model \cite{linial92,peleg00}.}
In the \LOCAL model of distributed computing a communication network is abstracted as an $n$-node graph $G=(V,E)$ of maximum degree $\Delta$. Nodes serve as computing entities and edges represent communication links. 
Nodes communicate with neighbors in synchronous rounds,
where in each round a node can perform arbitrary local computations and send one message of \emph{unbounded size} over each incident edge. The objective is to solve the problem at hand in the fewest rounds, e.g., with each node outputting its own color in a coloring problem.
In the \emph{distributed LLL} in the LOCAL model one generally assumes that the communication network is identical with the dependency graph $\cH$. This is motivated by the fact that in most applications of LLL the communication network and the dependency graph are in close resemblance and communication in $\cH$ can be simulated in the original communication network within a constant factor overhead in the round complexity. 
 In the \LOCAL model, there are $\poly\log\log n$ LLL algorithms for certain special types of LLLs, e.g., for LLLs with small maximum degree $\Delta\leq \poly\log\log n$ \cite{FGLLL17,RG20}, for LLLs with very strong LLL criteria \cite{Davies23}, or for  LLLs satisfying some additional technical properties \cite{GHK18}. 

The \CONGEST model is identical to \LOCAL, except for the important difference that messages are restricted in size; each message can only contain $O(\log n)$ bits, which fits only a constant number of node identifiers. As a result, one has to be much more careful when modeling distributed LLL instances and precise on which event and variable are simulated by which node of the communication network.
 We illustrate these challenges with a key example central to this paper.

\myparagraph{Example LLL (Slack Generation).} One prime application in our work is the slack generation LLL for graph coloring. The \emph{slack} of a node in a partial coloring of a graph is the number of colors that are available to it minus its uncolored degree in the graph. If every node has (positive) slack, the respective coloring problem can be solved via a sequential greedy algorithm and there are also $\poly\log\log n$-round distributed algorithms, even under the presence of bandwidth restrictions \cite{HKNT22,HNT22}. For the rest of this example, assume a sparse $\Delta$-regular graph $G=(V,E)$, i.e., a graph in which every node $v\in V$ has many non-edges in its neighborhood, that is, $G[N(v)]$ is far from being a clique. Such graphs can be colored with $\ll\Delta$ colors, but this cannot be done greedily as nodes may not have slack. We show that coloring some of the nodes such that every remaining node has slack (forming an easy greedily-solvable residual problem)  can be modeled as an LLL (see \cref{alg:slackgen} for the random process and \Cref{lem:slackgen-custom-simple,lem:slackgen-custom} for the formal statements):
\begin{quote}
    \emph{Each node is activated with constant probability and each activated node picks a random candidate color. Nodes with no neighbor with the same candidate color keep their color, and otherwise relinquish it.}
\end{quote}
A node $v$ gets slack in this process if two of its neighbors happen to get colored with the same color, as in that case $v$ only loses one available color from its color palette but two competing uncolored neighbors. Now, introduce a bad event for each node that holds if $v$ does not have slack after this process. Despite many dependencies between the final colors of different nodes, one can show that this forms an LLL, but the bandwidth constraints of the \CONGEST model make it extremely challenging to design efficient algorithms for this LLL. The bad event of a node $v$ depends on the randomness of nodes in its two-hop neighborhood, as these nodes' color choices determine which colors are finally retained in $v$'s neighborhood. Under bandwidth restrictions, $v$ can't learn about the candidate colors of all of these nodes. To exploit parallelism, all existing sublogarithmic-round distributed LLL algorithms gradually and carefully set more and more of the variables. To steer decisions on future variables in these processes, it is pivotal that event nodes learn \underline{all the information} about partial assignments of their variables. Hence none of these algorithms work efficiently in the \CONGEST model.\footnote{The only exception is the implementation of \cite{FGLLL17} in \cite{MU21} which is efficient for small maximum degree $\Delta$---the paper is formulated for $\Delta=O(1)$, but remains efficient for $\Delta$ up to $\poly\log\log n$---, as one can still learn the required information for a single step of their algorithm in $\poly \Delta$ rounds.}

\subsection{Our Contributions: LLL solvers}
\label{sec:ourcontrib}
\begin{tcolorbox}
First, as a conceptual contribution,  we formalize LLL problems in the \CONGEST model. 
\end{tcolorbox}
While it is straightforward to assume that each event and variable is simulated by some node of the communication network, it is \emph{a priori} unclear what knowledge nodes need regarding their events/variables. We reduce that knowledge or ability to a few simple primitives that can typically be implemented efficiently. Primarily, nodes need to be able to sample variables according to their distribution, measure how bad a partial solution is for events, and perform certain restricted communication between events and their variables. We refer to these LLLs as \emph{simulatable} (see \Cref{def:simulatability}). Note that the slack generation LLL as presented is not directly simulatable as one cannot measure the quality of partial variable assignments.

\subsubsection{Disjoint variable set LLLs}~
\begin{tcolorbox}
As our second contribution, we present an LLL algorithm for the setting where intuitively the variables of each bad event are split into two sets, and a good assignment for at least one of the two variable sets is sufficient to avoid the bad event (formal statement in \Cref{thm:LLLTwoSets}). 
\end{tcolorbox}

These LLLs appear frequently, e.g., when solving coloring LLL problems one may have colors from two different color spaces available. 
 Another simple example is given by the sinkless orientation problem whose objective is to orient the edges of a graph such that every node has at least one outgoing edge \cite{LLL_lowerbound,GS17}. The probability that a degree $\Delta$ node has no outgoing edge is upper bounded by $2^{-\Delta}$ when orienting the edges randomly, proving that this is an LLL. Now, one can use the splitting algorithm from \cite{HMN22} to split the edges into two sets such that every node has roughly $\Delta/2$ edges incident in each set, aligning the problem with our LLL framework and, to the best of our knowledge, yielding the first (published) CONGEST algorithm for the problem. 

 \noindent\textbf{Our approach in a nutshell.} Our algorithm for disjoint variable set LLLs  uses the influential shattering technique \cite{Beck1991,BEPSv3}. 
 First, we flip the variables in the first set at random according to their distribution to ensure that most of the events are avoided. A standard analysis shows that this \emph{shatters} the graph into small connected components---think of components of size $N=\poly\log n$. Then, we use the second set of variables to avoid the remaining bad events. In the \LOCAL model, the latter can be done with the known deterministic LLL algorithm from \cite{FGLLL17,RG20}  applied on all components independently and in parallel. This \emph{post-shattering phase} runs in $\poly\log N=\poly\log\log n$ rounds, exploiting the components' small sizes.  Our core technical contribution is an algorithm for the post-shattering phase in the \CONGEST model. A deterministic LLL \CONGEST algorithm is not known. Also, randomized algorithms are insufficient since their failure probability is $1/\poly(N) = 1/\poly(\log n)$ on each component, and almost surely one of the possibly many components will fail.  Instead, we run $\Theta(\log n)$ independent executions of a randomized algorithm in parallel to amplify the error probability. That, however, places an additional burden on the bandwidth and becomes the central challenge to overcome. We leverage the small component size for coordination and information learning in a more efficient manner, e.g., by using significantly smaller ID spaces so that a single CONGEST message can encode $\Theta(\log n)$ IDs. There are several technical details that we spare in this introduction. For example, we cannot set all variables of a small component in one go, and partially setting only some of the variables comes with the responsibility to ensure that there even exists a feasible assignment for the remaining unset variables.

\subsubsection{Binary LLLs with low risk}
\label{sec:introSampling}
In the absence of an alternative set of variables, the basic approach is to randomly sample all the variables, and then retract the variables around the events that fail under the initial assignment. The aim is then to solve another LLL on the subinstance induced by the retracted variables, redefining the events in terms of their marginal probability in the new instance, namely the probability that they hold conditioned on the assignment to the variables that are fixed.
The good news is that this instance would be small (poly-log size), due to shattering, so we can afford to apply more powerful LLL solvers.
The bad news is that this approach alone is seldom sufficient by itself since events that previously were not failing may now become highly unsatisfied by the retractions of adjacent events and there may not even exist an assignment of the retracted variables that avoids all bad events. Recursive retractions may lead to long chains, leading to at least $\Omega(\log n)$ rounds. The challenge is then how to limit retractions while ensuring a low conditional probability of bad events.

We treat a class of LLL with binary variables that occurs frequently in \emph{sampling}, where the sampled nodes are \black and the others \white.
Our approach is to perform a second round of retractions but only to the \white variables. 
We bound from above the marginal probabilities of the shattered instance in terms of a parameter that we call \emph{risk}. 
 \begin{tcolorbox}
As our third contribution, we present a bandwidth- and time-efficient algorithm for simulatable binary LLLs with low \risk (formal statement in \Cref{thm:promiseLLL}). 
\end{tcolorbox}

One example of such a problem is the sampling of a subset $S$ of the nodes of a sparse graph $G$ such that every node $v$ of the graph has few neighbors in $S$ but $G[N(v)\cap S]$ proportionally preserves the sparsity of $v$, i.e., the number of non-edges in its neighborhood. 
This \emph{Degree-Sparsity-Sampling problem (DSS)} is also essential to our coloring results, where we sample two (or more) such sets that can serve as alternative sets of variables for the slack generation LLL, effectively enabling us to solve the slack generation problem via our first LLL solver. 
In \Cref{sec:tecoverview}, we explain the details of why the problem fits our LLL framework, including its non-trivial simulatability. Its formal solution is presented in \Cref{sec:coloringSparseAndSlackGeneration}. This example illustrates also that 
our LLL solvers are most powerful when used in tandem.  In \Cref{sec:applications}, we present several additional examples of problems (and schemas of problems) that can be solved efficiently with our LLL algorithms in \CONGEST. 
Additionally, we show in \Cref{sec:applications} that any LLL that can be solved by the main LLL algorithm of \cite{GHK18} has low risk and hence, in can also be solved with our framework in \LOCAL.

 We point out to the knowledgeable reader that the criteria of our LLL solvers are crucially in between polynomial and exponential whenever $\Delta \ge \poly\log\log n$. This is the main regime of interest as there are $\poly\log\log n$-round LLL solvers for smaller $\Delta$ that work with a polynomial criterion and in \CONGEST \cite{MU21}, and often the bounds on the error probability in LLLs turn into with high probability guarantees for larger $\Delta$. 
 
\subsection{Our Contribution: Coloring Sparse and Triangle-Free Graphs}  
\label{sec:ApplicationsIntro}
Graph coloring is fundamental to distributed computing, as an elegant way of breaking symmetry and avoiding contention, and was in fact the topic of the original paper introducing the \LOCAL model \cite{linial92}. The typical setting that has been extensively studied is coloring a graph with $\Delta+1$ colors; in the centralized setting, such a coloring can be computed via a simple greedy algorithm. Importantly for the distributed setting, any partial solution to the problem can be extended to a full solution without ever needing to revert any coloring choice, a property that evidently does not hold when coloring with fewer colors and inherently makes it much more difficult to color large parts of the input graph in parallel. Sparse graphs admit colorings with $\ll\Delta$ colors, and logarithmic-time \LOCAL algorithms are known, e.g., to color triangle-free graphs with $\ll \Delta$ colors \cite{CPS17}.

We use our LLL algorithms to improve upon this result in two ways: First, our runtime is exponentially faster, second, in contrast to the prior algorithms our algorithms work in the CONGEST model. We summarize our results in the following theorem. Recall, that a node is \emph{sparse} if there are many non-edges in its induced neighborhood (see \Cref{sec:coloringSparseAndSlackGeneration} for the precise sparsity requirement).

\begin{tcolorbox}
There is a randomized \CONGEST algorithm that w.h.p.\ colors any triangle-free graph with $\Delta-\Omega(\Delta)$ colors and any locally sparse graph with $\Delta-\poly\log\log n$ colors. The algorithms run in $\poly\log\log n$ rounds (formal statements in \Cref{thm:sparseColoring,thm:triangleFreeColoring}).  
\end{tcolorbox}

\subsection{Further Related Work}
Due to its local nature, LLLs are a powerful tool in distributed computing, with several papers tackling its distributed complexity, e.g., \cite{brandt2016LLL,CPS17,FGLLL17,BMU19,BGR20,GHK18,CP19,CHLPU20,RG20,MU21,Davies23}. The prime characteristic of a (symmetric) LLL is its \emph{LLL criterion} that is the relation of the dependency degree $d$ (the maximum degree in $\cH$), and the global upper bound $p$ on the probability for each bad event to hold. The original lemma of \lovasz and Erd\"os \cite{LLL73} showed the existence of a feasible assignment as long as $e p (d+1) < 1$ holds. 
A simplified summary is that essentially all LLLs of interest ($(1+\eps)epd<1$ is required) can be solved in $O(\log n)$ rounds of \LOCAL \cite{MoserTardos10,CPS17}, while superfast $\poly(\log\log n)$ algorithms are only known for a special class of "(near) exponential" LLLs ($p2^{O(d/\poly\log\log n)}<1$) or  \cite{Davies23} or very low-degree (i.e., $\poly(\log\log n)$-degree) graphs under polynomial LLL criteria ($pd^{32}<1$) \cite{FGLLL17}, or when events satisfy certain additional robustness conditions \cite{GHK18}. See \cite{RG20,GHK18} for deterministic algorithms for polynomial LLLs. If the strong condition  $p<2^{-d}$ holds LLLs can be solved deterministically in $\poly\Delta +O(\log^*n)$ rounds \cite{BMU19,BGR20}.  In \CONGEST, the only algorithms known are randomized and for low-degree graphs \cite{MU21} and for "vertex splitting" problems \cite{HMN22}.

There are countless publications on the classic topic of distributed graph coloring with $\Delta+1$ colors focusing on different aspects of the problem, e.g., for coloring small degree graphs efficiently \cite{barenboim15,BEG17,FHK,MT20,M21,FK23}, for efficient deterministic coloring in LOCAL \cite{RG20,GK20,GG23} and in CONGEST \cite{BKM20,GK20}, and for randomized coloring algorithms in LOCAL \cite{johansson99,hsinhao_coloring,CLP20,HKNT22} and in CONGEST \cite{HKMT21,HN21,HNT22}. See also \cite{barenboimelkin_book} as a great resource covering many early results on the topic. 

\textit{Coloring with fewer colors:}
Recently, highly involved algorithms were designed to color non-clique graphs with $\Delta\geq 3$ with one fewer color, that is, with $\Delta$ colors \cite{GHKM18,FHM23}. The resulting designed $\poly\log \log n$ algorithm works in the \LOCAL model, but the algorithm inherently does not work in the \CONGEST model as one subroutine is based on learning the full topology of $\omega(1)$-diameter subgraphs. 

Chung, Pettie, and Su \cite{CPS17} give a \LOCAL algorithm for $\Delta/k$-coloring triangle-free graphs
for any $k = O(\log \Delta)$, running in $O(\log n)$ time (faster for very large $\Delta)$.
A much more constrained notion of sparsity is the \emph{arboricity} of a graph, that is, the number of forests into which one can partition the graph. For any constant $\eps>0$, there is a $O(\log n)$-round deterministic algorithms to color graphs with arboricity $\alpha$ with $O(2+\eps)\alpha$ colors, and it is known that the runtime is tight, even for randomized algorithms \cite{BE2010}.

\myparagraph{Notation.}
For a graph $G=(V,E)$ and two nodes $u,v\in V$ let $\dist_G(u,v)$ denote the length of a shortest (unweighted) path between $u$ and $v$ in $G$. For a set $S\subseteq V$ we denote $\dist_G(v,S)=\min_{u\in S}\dist_G(v,u)$.
For an integer $k\geq 0$ and a node $v\in V$ of a graph $G=(V,E)$ let $N_G^k(v)=\{u\in V : \dist_G(v,u)\leq k\}$. For a set $S$, we define $N^k_G(S)=\bigcup_{v\in S}N_G^k(v)$.

\subsection{Outline of the rest of the paper}
\Cref{sec:LLLdefinitions} contains our LLL formalization in CONGEST, followed by \Cref{sec:tecoverview} in which we present a technical overview of all results and techniques. 
\Cref{sec:samplingLLL} contains our LLL algorithm for LLLs with low risk. 
\Cref{sec:twoVariableLLL} presents our LLL algorithms working with two alternating sets of variables. 
Both of our LLL algorithms use the explained shattering framework that consists of a pre-shattering phase and a post-shattering phase. Our solution to the post-shattering phase is presented in \Cref{sec:CONGESTpostshattering}.
In \Cref{sec:applications}, we present several applications of our LLL algorithms.
In \Cref{sec:coloringSparseAndSlackGeneration}, we present our algorithm for DSS and our algorithms for coloring triangle-free and sparse graphs with few colors.

\vspace{-0.2cm}
\section{Distributed \lovasz Local Lemma (Definitions)}
\label{sec:LLLdefinitions}
In this section, we present our formalization of distributed LLL in the \CONGEST model.
\vspace{-0.2cm}
\subsection{Constructive \lovasz Local Lemma (LLL)} An instance $\cL=(\cV,\cB)$ of the \emph{distributed \lovasz local lemma (LLL)} is given by a 
a set $\cV=\{x_1,\ldots,x_{k_{\cV}}\}$ of independent random variables and a
family $\cB$ of "bad" events $\{\cE_1,\ldots,\cE_{k_{\cB}}\}$ over these variables. Let $\vbl(\cE)$ denote the set of variables involving the event $\cE$ and note that $\cE$ is a binary function of $\vbl(\cE)$.
The \emph{dependency graph} $\cH_{\cL}=(\cB, F)$ is a graph with a vertex for each event and 
an edge $(\cE,\cE') \in F$ whenever 
$\vbl(\cE)\cap \vbl(\cE')\neq\emptyset$. The \emph{dependency degree} $d = d_{\cL}$ is the maximum degree of $H_{\cL}$. 
We omit the subscript $\cL$ when the considered LLL is unambiguous. Additionally, we use the parameters $d_{\cE}=\max_{\cE\in \cB}|\vbl(\cE)|$ as the \emph{event degree} and $d_{\cV}=\max_{x\in \cV}\left|\{\cE\in \cB| x\in \vbl(\cE)\}\right|$ as the \emph{variable degree}. Define $p=\max_{\mathcal{E}\in\cB} \Pr(\mathcal{E})$ (or let $p$ simply be an upper bound on the term on the right-hand side). 
The \lovasz Local Lemma \cite{LLL73}  states that $\Pr(\cap_{\mathcal{E}\in \cB}\bar{\mathcal{E}})>0$ holds if $ep(d+1)<1$, or in other words, there exists an assignment to the variables that avoids all bad events.

 In the \emph{constructive \lovasz local lemma} one aims to compute such an \emph{feasible} assignment, avoiding all bad events.
 This is often under much stronger conditions on the relation of $p$ and $d$. The relation of $p$ and $d$ is referred to as the \emph{LLL criterion}. If $pd^c<1$ for some constant $c>1$ we speak of a \emph{polynomial criterion}, while if $p2^{d}<1$, we have an \emph{exponential criterion}. 
The problems we consider have criteria that are between the polynomial and exponential.

\subsection{Constructive Distributed \lovasz Local Lemma}
In the distributed setting, the LLL instance $\cL$ is mapped to a communication network $G=(V,E)$.
We are given a function $\ell: \cB\cup \cV\rightarrow V$ that assigns each variable and each bad event to a node of the communication network.  We assume that for each variable $x\in \cV$, the node $\ell(x)$ knows the distribution of $x$ including the range $\mathsf{range}(x)$ of the variable. We also say that node $\ell(x)$ \emph{simulates} the variable/event $x$. For  a vertex $v\in V$, we call $l(v)=|\ell^{-1}(v)|$ the \emph{load} of vertex $v$. The \emph{(maximum) vertex load} of an LLL instance is $l=\max_{v\in V}l(v)$.

In the constructive distributed LLL we execute a \LOCAL or \CONGEST algorithm on $G$ to compute a feasible assignment $\phi$. Afterwards, for each variable $x \in \cV$,  node $\ell(x)$ has to output $\phi(x)$. 

In general, the graph $G$ and the dependency graph $H_{\cL}$ do not have to coincide and $d$, $d_{\cE}$, $d_{\cV}$ and the maximum degree $\Delta$ of $G$ shall not be confused with each other. However, distances between events in $H_{\cL}$ and the corresponding nodes in $G$ should be closely related. 
\begin{definition}
A triple $(\cL,G,\ell)$ has \emph{locality} $\nu$ if $\dist_G(\ell(\cE),\ell(x))\leq \nu$ for all events $\cE$ of $\cL$ and variables $x\in\vbl(\cE)$. 
\end{definition}

\begin{observation}
\label{obs:distanceDependentEvents}
We have $\dist_G(\ell(\cE), \ell(\cE'))\le 2\nu$ for all dependent events $\cE,\cE'$.
\end{observation}

Note that the splitting problems of \cite{HMN22} modeled as LLLs have unit locality and unit distance between dependent events, which makes them particularly amenable to $\CONGEST$ implementations. 

If $(\cL,G,\ell)$ has small locality, then in the \LOCAL model we can perform all natural basic operations on events and variables of $\cL$ efficiently. Examples of such operations are testing whether an event holds under a variable assignment \cite{MoserTardos10,CPS17} or computing conditional event failure probabilities under a partial assignment \cite{FGLLL17,Davies23}. 
Any \LOCAL algorithm for the dependency graph can be simulated in the communication network with a multiplicative overhead of $O(\nu)$ in the round complexity, where $\nu$ is the locality of the LLL instance.

\myparagraph{(Partial) Assignments.}
An assignment of a set of variables $\cV$ is a function that assigns each variable $x\in \cV$ a value in $\mathsf{range}(x)$. 
 We use the value $\bot$ for variables that have not been set. A \emph{partial assignment $\phi$} of a set of variables $\cV$ is a function with domain $\cV$ satisfying $\phi(x)\in \mathsf{range}(x)\cup \{\bot\}$ for all $x\in \cV$. 
A partial assignment $\psi$ \emph{agrees} with another (partial) assignment $\phi$ if $\psi(x)=\phi(x)$ for all $x\notin \psi^{-1}(\bot)$, i.e., if all proper values assigned by $\psi$ match those of $\phi$.
For two partial assignments $\phi_1, \phi_2$ with $\bot \in \{\phi_1(x),\phi_2(x)\}$ for all variables $x$, let $\phi(x)=\phi_1\cup\phi_2$ be the assignment with $\phi(x)=\phi_1(x)$ whenever $\phi_1(x)\neq \bot$ and $\phi(x)=\phi_2(x)$, otherwise. For an event $\cE$ and a partial assignment with $\vbl(\cE)\cap \phi^{-1}(\bot)=\emptyset$, the term $\cE(\phi)$ states whether $\cE$ holds under $\phi$.

 A \emph{retraction} $\psi$ of a partial assignment $\phi$ is a partial assignment that agrees with $\phi$.
 We say that we \emph{retract} a variable $x$ of $\phi$ if we set its value to $\bot$ (formally this creates a new partial assignment).
For an event $\cE$ and a partial assignment $\phi$, we use the notation $\Pr(\cE \mid \phi)$ for the conditional probability over assignments with which $\phi$ agrees (the randomness is only over the variables in $\phi^{-1}(\bot)$).

\subsection{Simulatable Distributed \lovasz Local Lemma (CONGEST)}
\label{sec:simulatability}
In the \CONGEST model, a small locality of the function $\ell$ does not ensure that basic primitives can be executed efficiently. In \Cref{sec:tecOverviewsampleLLL}, we discuss the challenge of even evaluating the status of events via the example of the degree-bounded sparsity splitting problem (DSS).

A second challenge that we quickly touched upon appears from the need to make progress in large parts of the graph in parallel which requires us to set many (but not all) variables in parallel. The main difficulty is to ensure that we never run into an unsolvable remaining problem, that is, we need to ensure that the remaining variables can always be assigned values such that a feasible solution avoiding all bad events is obtained.  This is in stark contrast to problems like computing a maximal independent set or a $\Delta+1$-coloring in which any partial solution can always be completed to a solution of the whole graph. Thus, to make progress in large parts of the graph in parallel, we need to \emph{measure} how bad a partial assignment is for events that do not have all their variables set. Naturally, for a bad event $\cE$ and partial assignment $\psi$ this is captured by the conditional probability $\Pr(\cE\mid \psi)$. The \LOCAL model works of \cite{FGLLL17,GHK18,Davies23} implicitly compute these values. However, in \CONGEST it can be 
impossible to compute such marginal probabilities.  

In our simulatability definition (see \Cref{def:simulatability}) we do not require that conditional probabilities can be computed in the most general setting but only in the easier setting where we are given locally unique IDs from a small ID space. In the presence of little bandwidth, this helps significantly in some of the LLLs considered in this work and we believe that it will be helpful for other problems.

\medskip

 Next, we state the minimal assumptions that we require from an LLL instance.
 
\begin{definition}[simulatable]
\label{def:simulatability}
We say an LLL $(\cL,G,\ell)$ is \emph{simulatable} in $\CONGEST$ if each of the following can be done in $\poly\log\log n$ rounds:
\begin{enumerate}
    \item \textbf{Test:} Test in parallel which events of $\cL$ hold (without preprocessing).
    \item \textbf{Min-aggregation:} Given $1$-bit string in each event (variable), each variable (event) can simultaneously find the minimum of the strings for its variables (for its events).
    \medskip
    
    \item [] For the following items, it is sufficient if they hold in the setting that events and variables are given $O(\log\log n)$-bit IDs\footnote{In general for the whole LLL instance and for non-constant distances such identifiers do not exist, but our LLL algorithms only use the primitives in settings where they do exists and are available.} (that are unique within distance $4\nu$ in $G$):
    \item \textbf{Evaluate:} Given a partial assignment $\phi$, and partial assignments $\psi_1, \ldots, \psi_t$, $t = O(\log n)$, in which each variable knows its values (or $\bot$), each event $\cE$ of $\cL$ can simultaneously for all $1\leq i\leq t$ decide if
    \begin{equation*}
         \Pr(\cE\mid \psi_i)\leq \alpha \Pr(\cE\mid \phi)
     \end{equation*}
     holds , where $\alpha$ is a parameter known by all nodes of $G$. 
     \item \textbf{Min-aggregation:}  We can compute the following for $O(\log n)$ different instances in parallel: Given an $O(\log\log n)$-bit string in each event (variable), each variable (event) can simultaneously find the minimum of the strings for its variables (for its events).
\end{enumerate}
\end{definition}

Min-aggregation to variables allows events to retract their variables. Similarly, min-aggregation to events allows events to decide if they have a retracted variable. Min-aggregation with larger messages is used to find acyclic orientations of the dependency graph, a necessary step of the LLL algorithm of \cite{CPS17}. Due to the presence of smaller IDs, primitives (3) and (4) appear technical, but we emphasize that these are crucial to our \CONGEST solution. Further, it is unlikely that any sublogarithmic-time algorithm can go along without measuring the quality of partial assignments in one way or the other; with large IDs, the respective quality of partial assignments can provably not be checked in \CONGEST.

\vspace{-0.2cm}
\section{Technical Overview \& Technical Contributions}
\label{sec:tecoverview}
In \Cref{sec:tecOverviewtwoVariableLLL}, we give a technical overview of our disjoint variable set LLL solver. In \Cref{sec:tecOverviewsampleLLL}, we present the crucial ingredients of LLL solver for binary LLLs with low risk and explain how to use it for  the degree-bounded sampling problem. In \Cref{sec:tecOverviewPostshattering}, we present a condensed version of our approach to the LLLs arising in the post-shattering phase of our algorithms. 
In \Cref{sec:tecColoring}, we sketch how to use our LLL solvers to obtain our coloring results.

\subsection{Disjoint Variable Set LLLs}
\label{sec:tecOverviewtwoVariableLLL}
In this section, we present algorithms for \emph{disjoint variable set LLLs}, where we have two disjoint sets of variables $\cV_1, \cV_2$ available for each event. In fact, we consider events $\cE$ that can be written as the conjunction of two events $\cE_1,\cE_2$ where $\vbl(\cE_i)=\cV_i$ and $\Pr(\cE_i)\leq p$ holds for $i=1,2$. Note, that to avoid $\cE$ it is sufficient to avoid $\cE_1$ or $\cE_2$. Next, we sketch our algorithm. 
 
$\blacktriangleright$ We first sample all variables in $\cV_1$ according to their distribution (pre-shattering phase), and then move all non-avoided events (formally when $\cE_1$ is non-avoided) to the post-shattering phase. 
There we use the variables in $\cV_2$ to avoid the respective second events $\cE_2$. $\blacktriangleleft$
In contrast to our binary LLL solver, there are no retractions. For suitable $p$, the property $\Pr(\cE_1)\leq p$ ensures that the components in the post-shattering phase are of size $N=O(\log n\poly d)$, and the property $\Pr(\cE_2)\leq p$ ensures that each component in the post-shattering phase is an LLL.  The main focus of our work is the case where $d$ is at most polylogarithmic, in which case $N=\poly\log n$. In \LOCAL, we use \Cref{thm:deterministicLLLLOCAL}  to solve each small component in $\poly\log N=\poly\log\log n$ rounds. We obtain the following theorem.

\begin{restatable*}{theorem}{thmTwoVariableSet}
\label{thm:LLLTwoSets}
There are randomized \LOCAL and \CONGEST algorithms that in $\poly\log\log n$ rounds w.h.p.\ solve any \underline{disjoint variable set LLL} of constant locality $\nu$ with dependency degree $d\leq \poly\log n$ and bad event upper bound $p$. The \LOCAL algorithm requires $p<d^{-14}$ and the \CONGEST algorithm requires $p<d^{-(2+c_l)-(4c+12c_{\Delta}\nu)\log\log n}$, $l\leq d^{c_l}$, $\Delta\leq \log^c n$ for constants $c_l, c_{\Delta}\geq 1$, and simulatability. 
\end{restatable*}

The \CONGEST part of the theorem is more challenging for several reasons. We have already discussed the challenges regarding the evaluation of events and measuring progress for partial solutions. Another challenge is that shattering the dependency graph is not enough but the standard analysis only shatters the dependency graph. 
The issue is that, in \CONGEST, one cannot independently deal with different components of the dependency graph if the mapping of the events/variables of different components to the communication network overlaps. 

Thus, our solution relies on a stronger form of shattering in which we guarantee that after mapping the components of the dependency graph to the communication network, the components in the communication network remain small. This is one of the reasons why we require a stronger LLL criterion in \CONGEST; note that all of our applications satisfy this stronger criterion. 

The second difference between our \CONGEST and \LOCAL solutions lies in the post-shattering phase which we detail in \Cref{sec:tecOverviewPostshattering}.

\subsection{Binary LLLs with low risk}
\label{sec:tecOverviewsampleLLL}
In this section, we consider binary LLLs, that is, the range of the variables is $\{\mathsf{black},\mathsf{white}\}$.

\myparagraph{Outline of our binary LLL algorithm:}
 In our algorithm, each (original) event $\cE$ comes with an associated event $\mathsf{assoc}(\cE)$ (usually on the same variable set) that imposes stricter conditions, i.e., it is harder to avoid, but avoiding $\mathsf{assoc}(\cE)$ implies avoiding $\cE$. $\blacktriangleright$ We first flip all variables according to their distribution and then retract all variables around failing associated events. 
Then,  we perform a second round of retractions in which we only retract \white variables around those events that were affected by a (partial) retraction in the first round of retractions. 
We then form the residual LLL instance on the set of unset variables and all incident events. The probability of an event is now the conditional probability given the assignment to the unretracted variables.
We apply our post-shattering solver to this instance to produce the final solution. $\blacktriangleleft$

We introduce a new term, \emph{\risk}, which essentially upper bounds the conditional probability of a bad event to hold in the post-shattering phase under these promises. To show that it is small, we leverage the properties that hold for our retractions. 
Consider the four cases that can occur for an event: a) the event was \emph{unhappy}, i.e., $\mathsf{assoc}(\cE)$ occurred on the initial assignment, so all incident variables are retracted; b) the event was \emph{affected}: it was happy, but had incident variables retracted in the first round, so all incident white variables were retracted in second round; c) event was \emph{impacted}: it was not unhappy and not affected, but some incident white variables were retracted in second round (by another event), and d) the event was \emph{at peace}, with no incident variables retracted. We want to ensure that the conditional probability of the event is low in all these cases.

We say that an event pair $\cE, \mathsf{assoc}(\cE)$ testifies \emph{risk} $x$, if 
\begin{enumerate}
    \item $Pr(\mathsf{assoc}(\cE))\leq x$, and
    \item the marginal probability in cases a)--d) is at most $x$.
\end{enumerate}
Condition (1) is required to ensure that the probability of becoming unhappy is small such that the process shatters the graph into small components for the post-shattering phase. 
As $\cE$ implies $\mathsf{assoc}(\cE)$, condition (1) also ensures that the marginal probability for all events in case $a)$ is bounded above by $x$. 
Observe that in b)--d), the final assignment $\psi$ for the post-shattering phase was one derived from the initial one $\phi$ for which the event $\mathsf{assoc}(\cE)$ did not take place. The guarantees on $\psi$ for an event $\cE$ are that either no black variables were retracted from $\phi$ or all white variables were retracted from $\phi$. We say that in this case $\psi$ \emph{respects} the initial assignment $\phi$. Thus, condition (2) is equivalent to the following statement

\begin{center} $\triangleright$ $\Pr(\cE\mid \psi) \le x$ holds, for any assignment $\psi$ that respects some $\mathsf{assoc}(\cE)$-avoiding assignment.$\triangleleft$ \end{center}

In general, we show the following. 

\begin{restatable*}{theorem}{thmpromiseLLL}
\label{thm:promiseLLL}
There are randomized \LOCAL and \CONGEST algorithms that in $\poly\log\log n$ rounds w.h.p.\ solve any  LLL of constant locality $\nu$ with dependency degree $d\leq \poly\log n$ and \underline{\risk} $p$. The \LOCAL algorithm requires $p<d^{-14}$  and the \CONGEST algorithm requires $p<d^{-(4+c_l)-(4c+12c\nu)\log\log n}$, $l\leq d^{c_l}$, $\Delta\leq \log^{c_{\Delta}} n$ for constants $c_l,c_{\Delta}\geq 1$ and that the LLL is simulatable.
\end{restatable*}
The main difficulty with using \Cref{thm:promiseLLL} is to bound the risk of an LLL. As we prove in \Cref{sec:applications}, any LLL that can be solved with the main LLL algorithm of \cite{GHK18} has low risk. Hence, in the \LOCAL model, our algorithm subsumes the one of \cite{GHK18}. The issue is that it is difficult and technical to prove that an LLL fits the framework of \cite{GHK18} (also see \Cref{sec:applications} details). The core benefit of our approach is that it is superior for any binary LLL containing monotonically increasing events, that is, events that favor more nodes to be sampled. Examples are satisfying a minimum degree bound into a set of sampled nodes, or, as in the DSS problem, a minimum sparsity in the sample. In either case, the respective bad events are easier to avoid if we add more nodes to the sample. We prove the following lemma, formally proven in \Cref{sec:applications} (\Cref{L:monotone-incr}).

\vspace{-3\lineskip}
\begin{tcolorbox}[left=2pt, right=2pt]
\textbf{No risk Lemma.} 
\emph{The \risk of a monotone increasing event $\cE$ is  $\Pr(\cE)$ testified by $\assoc(\cE)=\cE$~. }
\end{tcolorbox}

 The name of the lemma stems from the fact that there is no additional risk from the conditional probabilities of the post-shattering phase. The conditional probability is identical to the probability of the event in the original LLL.
 Intuitively, the lemma holds, as any affected event $\cE$ retracts all of its \white incident variables, essentially, giving it free randomness for the post-shattering phase. The fact that some of its adjacent variables may remain \black can only make the situation better as the event prefers \black anyways and its conditional probability is upper bounded by its initial probability $\Pr(\cE)$. 

\newcommand{\cEmin}{\ensuremath{\cE^{\mathsf{min}}_v}\xspace}
\newcommand{\cEmax}{\ensuremath{\cE^{\mathsf{max}}_v}\xspace}

\smallskip

\textbf{Example 1 (degree bounded subgraphs):}  Let us see our framework in action with a first simple example. We are given a $\Delta$-regular graph and an integer $k$ (with $k \le \Delta/6$ and $k \gg \log \Delta$) and seek a subgraph $S$ such that each node has at least $k/3$ and at most $4k$ neighbors in $S$. For each node $v$ we have events \cEmin  and \cEmax that hold if the number of neighbors of $v$ in $S$ is less than $k/3$ and more than $4k$, respectively. For \cEmax the associated event asks to maintain a stronger of bound of $2k$ on the number of neighbors of $v$ in $S$. For \cEmin the associated event is the event \cEmin itself.

After an initial random sampling into a set $S$ with probability $q=k/\Delta$, each node has expected degree $k$ into $S$. Nodes with degree outside the range $[k/3, 2k]$ are unhappy, so all neighbors are retracted (both sampled and unsampled neighbors become \emph{undetermined}). Now, nodes that had an adjacent neighbor retracted go through the second round of retraction, with all \emph{unsampled} neighbors becoming undetermined. The post-shattering LLL is formed in terms of the undetermined nodes.
 Also see \Cref{fig:degreeBoundedSparsitySplitting} for an illustration of the four types of nodes (a)--d)).

\begin{figure}
\begin{tikzpicture}

\node (leftNode) at (0,0) [circle, draw] {$v$};

\foreach \i in {1,...,5} {
    \node (rightNode\i) at (2, -0.7*\i+2.1) [circle, draw, fill=gray!30] {?};
}

\foreach \i in {6,...,6} {
    \node (rightNode\i) at (2, -0.7*\i+2.1) [circle, draw, fill=white!30] {?};
}

\foreach \i in {1,...,6} {
    \draw (leftNode) -- (rightNode\i);
}

\end{tikzpicture}
\hspace{0.2cm}
\begin{tikzpicture}

\node (leftNode) at (0,0) [circle, draw] {$v$};

\foreach \i in {1,...,3} {
    \node (rightNode\i) at (2, -0.7*\i+2.1) [circle, draw, fill=gray!30] {\phantom{?}};
}
\foreach \i in {1,...,1} {
    \node (rightNode\i) at (2, -0.7*\i+2.1) [circle, draw, fill=gray!30] {?};
}

\foreach \i in {4,...,6} {
    \node (rightNode\i) at (2, -0.7*\i+2.1) [circle, draw, fill=white!30] {?};
}

\foreach \i in {1,...,6} {
    \draw (leftNode) -- (rightNode\i);
}
\end{tikzpicture}
\hspace{0.2cm}
\begin{tikzpicture}

\node (leftNode) at (0,0) [circle, draw] {$v$};

\foreach \i in {1,...,3} {
    \node (rightNode\i) at (2, -0.7*\i+2.1) [circle, draw, fill=gray!30] {\phantom{?}};
}

\foreach \i in {4,...,6} {
    \node (rightNode\i) at (2, -0.7*\i+2.1) [circle, draw, fill=white!30] {\phantom{?}};
}

\foreach \i in {5,...,6} {
    \node (rightNode\i) at (2, -0.7*\i+2.1) [circle, draw, fill=white!30] {?};
}

\foreach \i in {1,...,6} {
    \draw (leftNode) -- (rightNode\i);
}

\end{tikzpicture}
\hspace{0.2cm}
\begin{tikzpicture}

\node (leftNode) at (0,0) [circle, draw] {$v$};

\foreach \i in {1,...,3} {
    \node (rightNode\i) at (2, -0.7*\i+2.1) [circle, draw, fill=gray!30] {\phantom{?}};
}

\foreach \i in {4,...,6} {
    \node (rightNode\i) at (2, -0.7*\i+2.1) [circle, draw, fill=white!30] {\phantom{?}};
}

\foreach \i in {1,...,6} {
    \draw (leftNode) -- (rightNode\i);
}

\end{tikzpicture}
\caption{An illustration of the cases a)--d) that can appear in the post-shattering phase of the degree-bounded subgraph problem. Note that the illustration is only schematic and such a tight example with $\Delta=6$ does not satisfy any LLL criterion. The colors represent the variable assignments after the initial sampling. The question mark indicates that the respective variable got retracted and participates in the post-shattering phase. In a), the vertex $v$ got an extremely bad split and retracted all of its incident variables. In b), $v$ is affected by a retraction from a type a) node, and hence retracts all of its incident \white variables. In c), we see a node that was neither unhappy nor affected, but has some of its \white variables retracted by some other node of type c). The node in d) is happy and does not undergo any retractions. It does not participate in the post-shattering phase. W.h.p.\ the bulk of the nodes are of type d).\vspace{-0.4cm}}
\label{fig:degreeBoundedSparsitySplitting}
\end{figure}
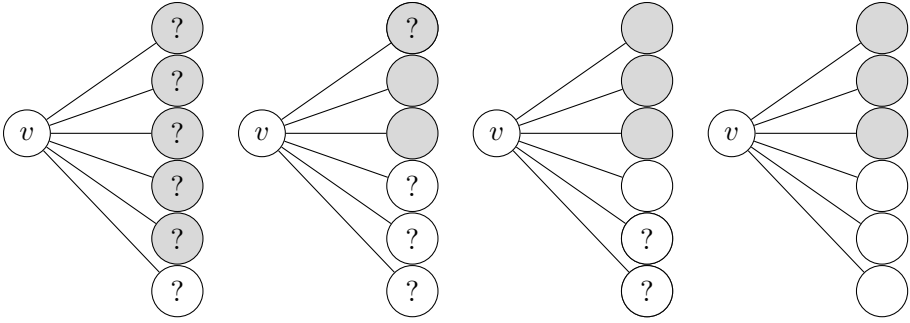

We reason that the risk of all events is small. \cEmin is monotonically increasing as it favors more nodes in the sample and by the No Risk Lemma its risk equals $Pr(\cEmin)$ which is upper bounded by $\exp(-\Theta(k))$ by a Chernoff bound. 
Similarly, a Chernoff bound shows that the probability of $\mathsf{assoc}(\cEmax)$, asking for at most $2k$ sampled neighbors, is at most $\exp(-\Theta(k)) \ll 1/\poly(\Delta)$, proving (1). This implies that the process shatters and we obtain small unsolved components in the post-shattering phase. 
To argue that the post-shattering phase indeed is a solvable LLL, what is left to prove is that $\cEmax$ has risk $x=1/\poly(\Delta)$;  we already bounded $Pr(\assoc(\cEmax))$ for proving (1). 

To prove (2), let us derive the conditional probabilities for all four types of nodes, and thereby the risk. The unhappy node has all its incident variables remain in the post-shattering instance, by Chernoff bound the conditional probability of $\cEmax$ is $\exp(-\Theta(k)) \ll 1/\poly(\Delta)$. The affected node has at most $2k$ incident nodes set (only the black ones) and therefore at least $\Delta - 2k \ge 2\Delta/3$ incident variables undetermined. By Chernoff, adding more than $2k$ black incident variables is highly unlikely. Hence, its final number of neighbors in $S$ will be larger than $4k$, with probability $\exp(-\Theta(k))$.
The impacted node has some white variables retracted, so its degree into $S$ may increase in the post-shattering step, but again by Chernoff, it will increase by more than $2k$ with probability $\exp(-\Theta(k))$. 
Finally, the events at peace already satisfy the requirement.
Hence, the conditional probabilities of $\cEmax$ are all at most $\exp(-\Theta(k)) \ll \Delta^{-32}$,  bounding the event's risk.

\smallskip

Not having to analyze how the \white-node-retraction affects the conditional probabilities of $\cEmin$, but instead relying on the No Risk Lemma is particularly helpful in our next example where the $\cEmin$ is replaced with a significantly harder to deal with \black-favoring event. 

\smallskip

\noindent\textbf{Example 2 (DSS):} Recall that we seek a subgraph $S$ such that for each node $v$, the graph $G[S \cap N(v)]$ has both low-degree and large sparsity. Let $\overline{m}_v$ denote the number of non-edges within $G[N(v)]$. Also see \Cref{fig:DSS}. 

Specifically, given a sampling probability $q$, the expected number of non-edges in $G[S\cap N(v)]$ is $q^2 \overline{m}_v$. 
Again, we have a bad event $\cEmax$ that holds if $G[S \cap N(v)]$ has more than $6q\Delta$ vertices, and a bad event $\cEmin$ that holds if $G[S \cap N(v)]$ has fewer than $q^2 \overline{m_v}/6$ non-edges.
The stricter respective events have the tighter bound of $2q\Delta$ on the degree into $S$, and the same lower bound on the number of non-edges.
Recall that nodes sampled into $S$ are \black, those not sampled are \white, while retracted nodes become \emph{undetermined}.

After this setup, the proof for bounding the respective risks is identical to the first example. $\cEmax$ and its associated event are of the same nature as in Example~1, and its risk can be bounded using identical arguments.  $\cEmin$ is a monotone increasing event, which again, by the No Risk Lemma has risk $\Pr(\cEmin)$. Bounding this probability is challenging in itself as it requires a concentration bound that handles dependencies (for reasoning about the number of non-edges in the sample), but the conclusion is the same: the risk is $O(\Delta^{-32})$, as desired. Here, the No Risk Lemma shows its power, as without it, one would have to deal with dependencies and hard-to-grasp conditional probabilities of partial assignments at the same time. 

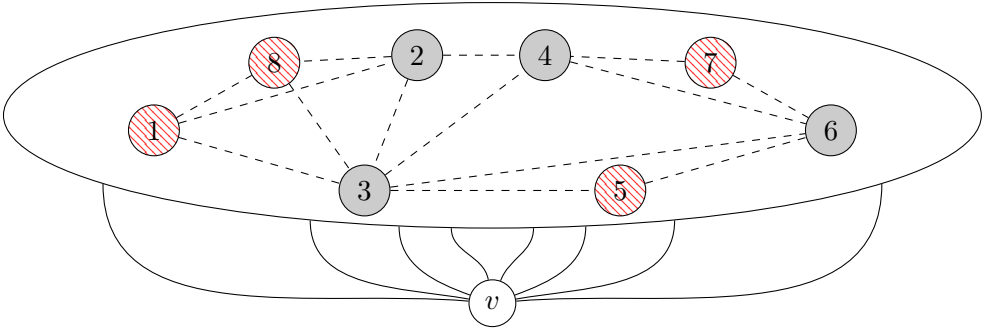
\begin{figure}
\begin{tikzpicture}
    \node (1) at (-1.0,2.3) [circle, draw, pattern=north west lines, pattern color=white!20!red] {1};
    \node (8) at (0.6,3.2) [circle, draw,  pattern=north west lines, pattern color=white!20!red] {8};
    \node (2) at (2.5,3.3) [circle, draw, fill=white!80!black] {2};
    \node (3) at (1.8,1.5) [circle, draw, fill=white!80!black] {3};
    \node (4) at (4.2,3.3) [circle, draw, fill=white!80!black] {4};
    \node (5) at (5.2,1.5) [circle, draw,  pattern=north west lines, pattern color=white!20!red] {5};
    \node (6) at (8,2.3) [circle, draw, fill=white!80!black] {6};
       \node (7) at (6.4,3.2) [circle, draw,pattern=north west lines, pattern color=white!20!red] {7};

\draw[dashed] (2) -- (3);
\draw[dashed] (1) -- (2);
\draw[dashed] (1) -- (3);
\draw[dashed] (3) -- (4);
\draw[dashed] (3) -- (5);
\draw[dashed] (3) -- (6);
\draw[dashed] (2) -- (4);
\draw[dashed] (5) -- (6);
\draw[dashed] (4) -- (6);
\draw[dashed] (7) -- (6);
\draw[dashed] (7) -- (4);
\draw[dashed] (2) -- (8);
\draw[dashed] (1) -- (8);
\draw[dashed] (8) -- (3);

\node(ellipse)[draw, ellipse, minimum width=13cm, minimum height=3cm] at (3.5,2.5){~};
 \node (node) at (3.5,0) [circle, draw, fill=white!80!white] {$v$};

\draw[-] (node) to[out=175,in=270] (ellipse.190);
\draw[-] (node) to[out=170,in=270] (ellipse.210);
\draw[-] (node) to[out=160,in=270] (ellipse.230);
\draw[-] (node) to[out=5,in=270] (ellipse.350);
\draw[-] (node) to[out=10,in=270] (ellipse.330);
\draw[-] (node) to[out=20,in=270] (ellipse.310);
\draw[-] (node) to[out=70,in=270] (ellipse.290);
\draw[-] (node) to[out=100,in=270] (ellipse.250);
\end{tikzpicture}
\vspace{-0.2cm}
\caption{An example of a node $v$ with $\Delta=8$ neighbors and two examples of sampled subsets (red patterned nodes, black nodes). The dashed edges are non-edges, that is, all other edges e.g., the edge $\{1,6\}$ are present in the graph. This neighborhood has $14$ non-edges out of the $\binom{8}{2}=28$ tentative edges.  While $v$ has only degree $4=\Delta/2$ into either of the two subsets, the red patterned subset would be a sampled subset with a small sparsity, as it only contains the single non-edge $\{1,8\}$. The black subset has larger sparsity, as it contains five non-edges $\{2,3\}, \{3,4\},\{2,4\}, \{4,6\},$ and $\{3,6\}$. The number of non-edges in a sampled set is not a linear function of the nodes' sampling status.\vspace{-0.5cm}}
\label{fig:DSS}
\end{figure}

 \textit{Simulatability in \CONGEST:} Another crucial benefit of our framework is that one can easily mix and mingle events. This is pivotal to ensure that the DSS is simulatable. Given some sample $S$, a node $v$ cannot even efficiently determine the number of non-edges in $G[N(v)\cap S]$, even if $\Delta$ is only polylogarithmic, as encoding the topology of the graph requires $\Omega(\Delta^2)$ IDs and a node can only receive $\Delta$ IDs of information per communication round.  But, a node can easily determine its degree $d_S(v)$ in $G[S]$, and reject the sample if its degree is too large. On the other hand, if the degree is small, also encoding the topology of the sampled subgraph in $v$'s neighborhood can be done more efficiently and hence the DSS-LLL becomes simulatable.  The No Risk Lemma is also helpful in that respect, as it shows that the associated event of a monotone increasing event is the event itself. Recall, that in the algorithm the associated event is needed to raise a red flag whenever the initial sampling goes wrong for the respective event. Prior work implicitly used involved associated events based on conditional probabilities which cannot be computed in CONGEST.

\smallskip

\vspace{-0.3cm}
\subsection{Post-shattering in CONGEST}
\label{sec:tecOverviewPostshattering}

In the post-shattering phase, we are given another LLL $\cL$ in a network with significantly fewer nodes, i.e., each component has at most $N=\poly\log n$ nodes. Hence, the original criterion can be restated as $p<d^{-\Omega(\log\log n)}=d^{-2\log N}$, while the bandwidth remains the original $\Theta(\log n)$. In the \LOCAL model, we can immediately solve this in $\poly\log\log n$ rounds via \Cref{thm:deterministicLLLLOCAL}, but that requires large bandwidth for gathering large parts of the graph at a single node.

$\blacktriangleright$ In our CONGEST algorithm, we first compute a network decomposition of the components into $C = O(\log N)$ collections consisting of $O(\log^2 N)$-diameter clusters \cite{RG20,GGR20} with a distance $k=2\nu$ between clusters in the same collection (recall, $\nu$ is the locality of the LLL). 
Then, we iterate through the collections in sequence. Before each iteration $i$, we have a partial assignment $\phi_{i-1}$, formed by the assignments made in previous iterations, 
and formulate a new LLL $\cL_i$ on the unset variables of nodes in the $i$-th collection. The bad events of $\cL_i$ ensure that after setting these variables the conditional probability of each original event (of $\cL$) increases at most by a factor $d^2$. By Markov's inequality the probability that the increase is larger than $d^2$ if a subset of the variables of an event are sampled is at most $1/d^2$ (see \Cref{claim:LLLMarkov} for details). By induction over $i$, we maintain the invariant that after processing the $i$-th collection,
we have $\Pr(\cE\mid \phi_i)\leq p\cdot d^{2i}$ for each bad event $\cE$ of $\cL$. 
Thus, at the end $\Pr(\cE\mid \phi)\leq p d^{2C} < 1$, if we assume the criterion $p < d^{-2C}$. Since all variables have been fixed by $\phi = \phi_{C+1}$, the event $\cE$ is avoided under the final assignment $\phi$.

To solve $\cL_i$ on each cluster (and thus each collection), we run $O(\log n)$ parallel instances of the LLL algorithm of \cite{CPS17}. Each of them succeeds (avoids all bad events of $\cL_i$) with probability $1-1/N\geq 1/2$.
Thus, at least one of these instances succeeds w.h.p. To determine a successful assignment, we use bitwise aggregation, utilizing the small cluster diameter. $\blacktriangleleft$

Simulatability (\Cref{def:simulatability}) is the key to solving these instances in parallel in \CONGEST. It ensures that each of the steps of the \cite{CPS17} algorithm of all instances in parallel can be implemented fast enough. The full details are in \Cref{sec:CONGESTpostshattering}. We also need to efficiently communicate between events and their variables, e.g., to resample variables. Note that many of our LLLs
can only perform these steps efficiently after we compute smaller locally unique node IDs from an ID space of size $\poly N$, which only requires $O(\log\log n)$ bits per ID.

Lastly, we want to remark that the idea of amplifying probabilities by running several instances of an algorithm in parallel has been used before, but in significantly simpler settings \cite{HMN22,ghaffari19_MIS}.

\vspace{-0.2cm}
\subsection{Coloring Sparse Graphs}
\label{sec:tecColoring}
Recall, that providing slack to sparse nodes by partially coloring the graph can be modeled as an LLL. Once uncolored nodes have slack, we can complete their coloring by a simple $deg+1$-list coloring procedure from prior work (this brings us back to the greedy coloring regime that is well-understood), \cite{HNT22}. Unfortunately, as discussed, the slack generation LLL is not simulatable and cannot be tackled easily in the CONGEST model.  To provide slack to nodes in CONGEST, we use the DSS problem to compute two (or more)  degree-bounded sparsity-preserving sets $S_1$ and $S_2$.
Having these degree-bounded sets with many non-edges in each neighborhood has several benefits. First of all, if we only color nodes in the degree-bounded sets, the slack generation problem becomes simulatable. Secondly, we also partition the color space into two linearly-sized sets that are then used for coloring $S_1$ and $S_2$, respectively. As every node can obtain slack from coloring nodes in either set, this effectively splits the slack generation LLL variables into two sets and aligns with our two disjoint sets LLL solver.  We obtain a slack generation algorithm with runtime $\poly\log\log n$.

\vspace{-0.2cm}

\section{Binary LLLs with low \Risk}
\label{sec:samplingLLL}
In this section, we consider binary LLLs, that is, the range of the variables is $\{\mathsf{black},\mathsf{white}\}$.

For an event $\cE'$, let $\mathsf{Retract}(\cE')$ consist of all assignments $\psi$ that are a retraction of some full assignment $\phi$ under which $\cE'$ is avoided. Let $\mathsf{Respect}(\cE')\subseteq \mathsf{Retract}(\cE')$ be the set of assignments $\psi$ that additionally have the guarantee that either all \white variables under $\phi$ in $\vbl(\cE')$ are retracted, i.e., $\vbl(\cE')\cap\phi^{-1}(\white)\subseteq \psi^{-1}(\bot)$, or no \black variables under $\phi$ in  $\vbl(\cE')$ are retracted, i.e., $\vbl(\cE')\cap \phi^{-1}(\black)\cap \psi^{-1}(\bot)=\emptyset$. 

\begin{restatable}[\risk]{definition}{defPromiseretractionCost}
\label{def:promiseretractionCost}
We say that an event $\cE'$ \emph{testifies}  \risk $x$ for some event $\cE\subseteq \cE'$ if 
\begin{align}
    \max\big\{\Pr(\cE'),\max_{\psi\in \mathsf{Respect}(\cE')}\{\Pr(\cE\mid \psi)\}\big\}\leq x~.
\end{align}
The \emph{risk} of an event $\cE$ is the smallest risk testified by some event $\cE'\supseteq \cE$~.
\end{restatable}

The bound on $\Pr(\cE')$ will ensure shattering. The bound on $\max_{\psi\in \mathsf{Respect}(\cE')}\{\Pr(\cE\mid \psi)\}$ upper bounds the marginal probabilities of event $\cE$ in the post-shattering phase. The intuition for the condition $\cE'\supseteq\cE$ is that we want $\cE$ to be avoided if an event is at peace during the whole process.

The next definition captures the binary LLLs that we deal with in this section. 

\begin{definition}[binary LLLs with low risk]
\label{def:samplingLLL}
A \emph{binary LLL with risk $p$} consists of the following:
\begin{itemize}
\item $\cV$ a set of binary independent random variables with range $\{\black,\white\}$~, 
\item $\cB$ a set of events over $\cV$ with \risk at most $p$ and $\Pr(\cE)\leq p$  for all $\cE\in \cB$~,
\item For each event $\cE\in \cB$ an associated event $\mathsf{assoc}(\cE)$ testifying its \risk.
\end{itemize}
The dependency degree $d$ is the maximum degree of the dependency graphs induced by all events.
The goal is to compute an assignment of the variables in $\cV$ such that all events in $\cB$ are avoided. 
We extend the definition of simulatability of a binary LLL with low risk and also require that the associated events can also be evaluated in $\poly\log\log n$ rounds on any assignment of the variables. 
\end{definition}

\begin{remark}
Note that the risk of an event as given in \Cref{def:promiseretractionCost} minimizes over all possible associated events and as such may not be easily computable. Hence, in \Cref{def:samplingLLL}, we require that the respective associated events are known in a simulatable manner to the nodes of the network.
\end{remark}

We prove the following theorem.

\thmpromiseLLL

We next present our algorithm for \Cref{thm:promiseLLL} that uses the shattering technique. The \CONGEST version of the algorithm requires  relies on the post-shattering algorithm from \Cref{sec:CONGESTpostshattering}.

\myparagraph{Algorithm:} Consider a binary LLL as in \Cref{def:samplingLLL}.
\begin{itemize}
    \item \textbf{Initial sampling (Step~1):} Sample all variables in $\cV$ according to their distribution. 

    Let $\phi$ be the resulting assignment. 
    \item \textbf{Retraction I:}  For each $\cE\in \cB$ for which $\mathsf{assoc}(\cE)$ holds under $\phi$, retract all incident variables,
    \item \textbf{Retraction II:} For each $\cE\in \cB$ with an unset variable, retract all incident $\mathsf{white}$ variables, 

    Let $\psi_{\pre}$ be the resulting partial assignment. 
        \item \textbf{Post-shattering:} Set up the following LLL problem consisting of all unset variables and their incident (marginal) events. 
        \begin{itemize}
        \item \textbf{Variables:} $\cV_{\post}=\psi_{\pre}^{-1}(\bot)$, with their respective original probability distribution, 
        \item \textbf{Bad Events:} $\cB_{\post}=\{\cE|\psi_{\pre} : \cE\in \cB, \vbl(\cE)\cap \cV_{\post} \neq \emptyset\}$ \        
        \item $\ell_{\post}(x)=\ell(x)$ for all $x\in \cV_{\post}$ and $\ell_{\post}(\cE')=\ell(\cE)$ for all $\cE'\in \cB_{\post}$. To simplify the notation, we refer to $\ell_{\post}$ as $\ell$.
        \end{itemize}
        
        We compute an assignment $\psi_{\post}$ of all variables in $\cV_{\post}$ avoiding all events in $\cB_{\post}$. In \LOCAL we use \Cref{thm:deterministicLLLLOCAL} and in \CONGEST we use our algorithm from  \Cref{lem:CONGESTpostshattering}.
        \item \textbf{Return} $\psi=\psi_{\pre}  \cup \psi_{\post}$~.
\end{itemize}

\myparagraph{Analysis of the post-shattering phase.}
We first show that the LLL formulation in the post-shattering phase indeed is an LLL. Afterwards we show that with high probability the dependency graph of this LLL consists of small connected components.

\begin{lemma}
\label{lem:sampleLLLpostShattering}
$\cL_{\post}$ is an LLL with bad event probability bound $p$ and dependency degree $d$. 
\end{lemma}
\begin{proof}
The dependency degree of the LLL is upper bounded by the $d$ in \Cref{def:samplingLLL}. 

Let $\phi$ be the assignment of the variables of Step~1. 
Let $\cE\in \cB_{\post}$ be an arbitrary event. If $\mathsf{assoc}(\cE)$ holds under $\phi$, then $\vbl(\cE)\subseteq \cV_{\post}$ and we obtain $\Pr(\cE \mid \psi_{\pre})=\Pr(\cE )\leq p$  by \Cref{def:samplingLLL}. 

If $\mathsf{assoc}(\cE)$ is avoided under $\phi$ we obtain $\Pr(\cE \mid \psi_{\pre})\leq p$  by the definition of the \risk (see \Cref{def:promiseretractionCost}), as $\psi_{\pre}\in \mathsf{Respect}(\mathsf{assoc}(\cE))$ and $\mathsf{assoc}(\cE)$ testifies the risk $p$.  
\end{proof}
Note that the statement in \Cref{lem:sampleLLLpostShattering} holds ``deterministically'', that is, regardless of the outcome of the random choices in Step~1 of the algorithm. 
Let $W=\{v\in V | \ell^{-1}(v)\cap(\cV_{\post}\cup \cB_{\post})\neq\emptyset\}$ be the set of nodes that have one of their events/variables participating in the post-shattering phase. Further, let $W'=\{v\in V :  \dist_G(v,W)\leq \nu\}$ be the nodes that are in distance at most $\nu$ from $W$. 

We obtain the following bounds for events, variables, and nodes to be part of the post-shattering. 
\vspace{-8\lineskip}
\begin{restatable}{lemma}{lemProbabilitiesLow}The following bounds hold. \label{lem:shatteringProbabilitieslowContraction}
\begin{itemize}
    \item For each event $\cE\in \cB$ we have $\Pr(\cE\in \cB_{\post})\leq p\cdot d^2$~,
    \item For each variable $x\in \cV$ we have $\Pr(x\in \cV_{\post})\leq  p\cdot d^2\cdot (d+1)$~,
    \item For each node $v\in V(G)$ we have $\Pr(v\in W)\leq p\cdot d^2 \cdot (d+1)\cdot l$~,
    \item For each node $v\in V(G)$ we have $\Pr(v\in W')\leq p\cdot d^2 \cdot  (d+1)\cdot l\cdot \Delta^{\nu}$.
\end{itemize}
For distinct $\cE$ and $\cE'$ the events $\cE\in \cB_{\post}$ and $\cE'\in \cB_{\post}$ are independent if $\dist_H(\cE,\cE')> 1$. For a node $v$ the event whether it is contained in $W$ or $W'$ only depends on the randomness at nodes $v'$ with $\dist_G(v,v')\leq 5\nu$ and $\dist_G(v,v')\leq 6\nu$, respectively. 
\end{restatable}
\begin{proof}[Proof Sketch, formal proof thereafter]
Each event or variable participating in the post-shattering phase has some event $\cE$ in its vicinity for which $\assoc(\cE)$ holds after the initial sampling. $\Pr(\assoc(\cE))$ is bounded by $p$ due to the bounded risk. The lemma follows with several union bounds over the respective sets of events and variables in multiple hop distance neighborhoods. 
\end{proof}
\begin{proof}[Proof of \Cref{lem:shatteringProbabilitieslowContraction}]
    Define the following sets of events and variables:
\begin{align}
    E_0& =\{\cE\in \cB :  \mathsf{assoc}(\cE)\text{ holds under }\phi\}, & V_0& =\bigcup_{\cE\in X_0} \vbl(\cE),\\
    E_1& =\{\cE\in \cB\setminus E_0 :  \vbl(\cE)\cap V_0\neq\emptyset\}, & V_1 & =\bigcup_{\cE\in X_1}\vbl(\cE)\setminus V_0,\\
    E_2 &= \{\cE\in \cB\setminus (E_0\cup E_1) :  \vbl(\cE)\cap V_1\neq\emptyset\} ~.
\end{align}
Observe that the set $E_0$ consists of all events that retract a variable in the first step of retractions and $V_0$ contains all the retracted variables. $E_1$ contains those events that are not contained in $E_0$ but depend on a retracted variable. The events in $E_1$ cause the retraction of their \white variables in the second round of retractions. The set $V_1$ consists of the variables that are retracted in that step and $E_2$ contains all events that are neither in $E_0$ nor in $E_1$ but are adjacent to a retracted variable. Altogether $E_0\cup E_1\cup E_2$ contains all events participating in the post-shattering phase and $V_0\cup V_1$ contains all variables participating in the post-shattering phase.

Fix an arbitrary event $\cE$. By \Cref{def:samplingLLL}, the probability that $\cE$ is contained in $E_0$ is at most $p$. Furthermore, $\cE$ cannot be contained in $E_0\cup E_1\cup E_2$ if none of the events that are within distance at most $2$ in the dependency graph $H_{\cL}$---note that this is the dependency graph of the original LLL and not the LLL in the post-shattering phase---are contained in $E_0$. There are at most $d^2$ such events, and with a union bound over we obtain $\Pr(\cE\in \cB_{\post})=\Pr(\cE\in E_0\cup E_1\cup E_2)\leq d^2 p$. 

Fix an arbitrary variable $x\in \cV$. In order for $x\in \cV_{\post}$ to hold, one of its $d_{\cV}$ adjacent events needs to participate in $\cB_{\post}$. We obtain $\Pr(x\in \cV_{\post})\leq  p\cdot d^2 \cdot d_{\cV}\leq d^2(d+1)$ with a union bound over these events.

Fix an arbitrary node $v\in V$. For $v\in W$ to hold, one of its $l(v)$ events/variables needs to be contained in $\cB_{\post}\cup \cV_{\post}$. We obtain $\Pr(v\in W)\leq p\cdot l(v) \cdot d^2 \cdot d_{\cV}\leq p\cdot l \cdot d^2 \cdot d_{\cV}$ with a union bound over all these events/variables. Similarly, $v\in W'$ only holds if there is a node $u\in N^{\nu}(v)$ with $u\in W$. We obtain $\Pr(v\in W')\leq p\cdot l \cdot d^2 \cdot d_{\cV}\cdot \Delta^{\nu}$ with a union bound over the $\Delta^{\nu}$ nodes in $N^{\nu}(v)$.

Note that the event $(\cE\in \cB_{\post})$ only depends on the random choices of variables of events in distance at most $2$ in $H$. Thus, $(\cE\in \cB_{\post})$ and $(\cE'\in \cB_{\post})$ for events $\cE,\cE'$ depend on distinct sets of variables if $\dist_H(\cE,\cE')>4$. 

For a node $v\in V$ the event $(v\in W)$ holds if one of the events/variables in $\ell^{-1}(v)$ participate in $\cV_{\post}\cup \cB_{\post}$. For a variable $x\in \ell^{-1}(v)$, the longest dependency chain is that $x\in \vbl(\cE_1)$, where $\cE_1$ cause a retraction of $x$ in the second round of retractions, and $\cE_1$ cause the retraction, as some other event $\cE_2$ retracted a variable $y\in\vbl(\cE_1)$ because $\cE_2$ was not avoided under the initial assignment $\phi$ of its variables $\vbl(\cE_2)$. Let $z\in \vbl(\cE_2)$. Then, the event $(v\in W)$ may depend on the randomness of $z$, that is, the largest distance is at most $\dist(v,\ell(z))\leq 4\nu$. For an event $\cE$ a similar reasoning upper bounds the dependence to $5\nu$. For $v\in W'$, this distance increase by an additive $\nu$ by the definition of $W'=N^{\nu}(W)$.

To obtain the bounds in the lemma we use that $d\geq d_{\cV}-1$ as all events incident to a variable are dependent. 
\end{proof}

We obtain that w.h.p.\ each connected component in the dependency graph of the LLL in the post-shattering graph contains few events, that is, the connected components of $H[\cB_{\post}]$ are small. This is sufficient for the post-shattering in the \LOCAL model as small locality $\nu$ implies that we can simulate any algorithm on these small components efficiently and in parallel in $G$.  

In the \CONGEST model, we require stronger guarantees. Observe, that in an LLL instance with load $l$ a vertex can simulate up to $l$ events/variables and it may be that these are not dependent on each other. So, even though the components in the dependency graph of the LLL in the post-shattering phase may be small, it may be that their projection via $\ell$  to the communication network may result in huge connected components of $G$. For an efficient post-shattering phase, we require that also the projections to the communication network remain shattered, as we prove in \Cref{lem:promiseSmallComponentContained}.

 Observe that two events in the dependency graph $H_{\cL_{\post}}$ of $\cL_{\post}$ are connected by an edge if they share a variable in $\cV_{\post}$. So, $H_{\cL_{\post}}\subseteq H[\cB_{\post}]$ where $H_{\cL_{\post}}$ does not contain an edge $\{\cE,\cE\}$ of $H[\cB_{\post}]$ if all variables in $\vbl(\cE)\cap \vbl(\cE')$ already have a value in $\psi_{\pre}$.

\vspace{-3\lineskip}
\begin{lemma}[shattering]
\label{lem:sampleLLLshattering}
If $p\leq d^{-22}$ holds, then with high probability each connected component of $H_{\cL_{\post}}\subseteq H[\cB_{\post}]$ contains at most $O(d^{8}\cdot \log n )$ nodes.

If $p<(d^2\cdot d_{\cV}\cdot l\cdot \Delta^{\nu})^{-1}\cdot \Delta^{-(4+12\nu)}$ holds, then with high probability the graph $G[W']$ consists of connected components of size $O(\log n \cdot \Delta^{6\nu})$. 
\end{lemma}
\vspace{-4\lineskip}
\begin{proof}
By \Cref{lem:shatteringProbabilitieslowContraction} and the shattering \Cref{lem:Shattering} (with $c_3=2, c_2=4, c_1=21$), we obtain that with high probability in $n$ the connected components of $H[\cB_{\post}]$ are of size at most $O(\log n\cdot d^8)$ if $p\leq d^{-22}$ holds. This high probability guarantee depends on the randomness of Step~1.

By \Cref{lem:shatteringProbabilitieslowContraction} and the shattering \Cref{lem:Shattering} (with $c_3=2, c_2=6\nu, c_1=4+24\nu$), we obtain that with high probability in $n$ the connected components of $G[W']$ are of size at most $O(\log n \cdot \Delta^{12\nu})$ if $\Pr(v\in W')=p(d^2\cdot (d+1)\cdot l\cdot \Delta^{\nu})\leq  \Delta^{-c_1}$ holds, only using randomness of Step~1. 
\end{proof}

\vspace{-5\lineskip}
\begin{lemma}
\label{lem:promiseSmallComponentContained}
The projection of each connected component of $H[\cB_{\cL_{\post}}]$ via $\ell$, including all incident variables in $\cV_{\cL}$ of the events in $\cB_{\cL_{\post}}$, is contained in a single connected component of $G[W']$.
\end{lemma}
\begin{proof}
Assume for contradiction that there are two nodes $w_1,w_2$ in different connected components of $G[W']$ that simulate dependent events $\cE_1,\cE_2$ of $\cL_{\post}$. By the definition of $\cL_{\post}$ we obtain that $w_1,w_2\in W$. Let $x\in \vbl(\cE_1)\cap\vbl(\cE_2)$. By the locality of $\ell$, we obtain $d(w_i,\ell(x))\leq \nu$ for $i=1,2$. But this would imply that $\ell(x)$ is in the same connected component of $G[W']$ as $w_i$ for both $i=1$ and $i=2$, which is a contradiction. The same holds for all variables of these events.
\end{proof}

\Cref{thm:promiseLLL} follows by plugging all lemmas in this section together, using that the LLL is simulatable for the algorithm's CONGEST implementation and verifying that the bound on the error probability $p$ in \Cref{thm:promiseLLL} is sufficient to actually apply \Cref{lem:sampleLLLshattering}. 

\begin{proof}[Proof of \Cref{thm:promiseLLL}] Step~1 and the retractions can be implemented in $O(\nu)$ rounds in the \LOCAL model. In the \CONGEST model, Step~1 and the first round of retractions can be implemented in $\poly\log\log n$ rounds if $(\cL,G,\ell)$ is simulatable. More detailed with the first primitive of \Cref{def:simulatability}, node $\ell(\cE)$  can test whether $\mathsf{assoc}(\cE)$ holds (note that we extended the definition of simulatability to associated events, see \Cref{def:samplingLLL}) and via the second primitive respective nodes can send a \emph{retract} command to the respective variables.  the first and second primitive in ). The second round of retractions is implemented as follows. First, each event determines whether it has a retracted variable by the aggregation primitive. Each retracted variable sends out a bit $0$, each variable with a value $\neq\bot$ sends out the bit $1$, and these are combined with the minimum operator. Thus, for each event $\cE\in \cB$ the node $\ell(\cE)$ can determine whether it has a retracted variable after the first round of retractions. If this is the case, the event sends the bit $0$ to all its variables indicating that \white variables should retract their value. Note that these commands can be combined (formally aggregated via the min operator) if a variable gets that retraction command from multiple of its events. With similar usage of these primitives for each event $\cE\in\cB$ the node $\ell(\cE)$  can determine whether the event $\cE$ should participate in the post-shattering phase, i.e., whether there is some $x\in \vbl(\cE)$ with $\phi_{\pre}(x)=\bot$. 

At the end of the algorithm, all events in $\cB$ are avoided under $\phi$ for the following reasons. Let $\cE\in \cB$ and let $S=\vbl(\cE)\cap \cV_{\post}$ the variables of the event whose assignment is computed in the post-shattering phase. If $S=\emptyset$, then $\mathsf{assoc}(\cE)$ is avoided under $\psi_{\pre}$ and we have $\psi(x)=\psi_{\pre}(x)$ for all $x\in \vbl(\cE)$. Then $\cE$ is avoided under $\phi$ as $\cE\subseteq \mathsf{assoc}(\cE)$ holds. If $S\neq \emptyset$, the assignment of the variables in $S$ in $\psi_{\post}$  together with $\psi_{\pre}$ avoids the event $\cE$ by the definition of the corresponding event in $\cB_{\post}$.

\Cref{lem:sampleLLLpostShattering} shows that $\cL_{\post}$ indeed is an LLL with dependency degree $d$ and upper bound $p$ on the bad event probability.

The application of \Cref{thm:deterministicLLLLOCAL} and \Cref{lem:CONGESTpostshattering} for solving $\cL_{\post}$ in the post-shattering phase is almost identical to the one in the proof of \Cref{thm:LLLTwoSets}. The respective lemmas need to be replaced with the lemmas of this section and \Cref{eqn:shatteringCondition} needs to be replaced with
 \begin{align}
     p<d^{-(4+c_l)-(4c+12c\nu)\log\log n}\leq  (d^2\cdot (d+1)\cdot l)^{-1}\Delta^{-(4+12\nu)}~, 
\end{align}
which holds due to the bound on $p$ from \Cref{thm:promiseLLL} and justifies that we can use \Cref{lem:sampleLLLshattering}.
 \end{proof}

\section{Disjoint Variable Set LLLs}
\label{sec:twoVariableLLL}

\begin{restatable}[Disjoint variable set LLLs]{definition}{defTwoVariable}
\label{def:twoVariableLLL}
Consider a set $\cV$ of random variables that is the union of two disjoint sets of variables $\cV_i, i=1,2$, and some set of events $\cE$ defined over $\cV$. This is a \emph{disjoint variable set LLL} with bad event probability bounded by $p$ if each event $\cE=\cE_1\cap \cE_2$ is the conjunction of two events $\cE_1, \cE_2$ where 
 $\vbl(\cE_i)=\cV_i$ and $\Pr(\cE_i)\leq p$ holds for $i=1,2$.
\end{restatable}
Note that for an event $\cE=\cE_1\cap \cE_2$ it is sufficient to avoid $\cE_1$ or $\cE_2$ to avoid the event $\cE$. In the case of a disjoint variable set LLL in the \CONGEST model, we extend the definition of simulatability and require that the event $\cE_1$ can be tested in $\poly\log\log n$ rounds on any partial assignment $\phi$ with $\bot\notin\phi(\vbl(\cE_1))$. The rest of this section is devoted to proving the following theorem. 

\thmTwoVariableSet

Our algorithm first samples all variables in $\cV_1$.
Events not avoided by the assignment to their variables in $\cV_1$ join the post-shattering phase together with their variables in $\cV_2$. The benefit is that connected components in the post-shattering are of polylogarithmic size which allows for faster algorithms, both in the \LOCAL model (\Cref{thm:deterministicLLLLOCAL}) and in the \CONGEST model under certain stronger LLL criteria (\Cref{sec:CONGESTpostshattering}).

\myparagraph{Algorithm:} Consider an LLL $\cL=(\cV,\cB,p,d)$ as in \Cref{def:twoVariableLLL}.

\begin{itemize}
    \item \textbf{Initial sampling (Step~1):} Sample all variables in $\cV_1$ according to their distribution. 

    Let $\psi_{\pre}$ be the resulting partial assignment. 
        \item \textbf{Post-shattering:} Set up the following LLL instance $\cL_{\post}=(\cV_{\post},\cB_{\post},\ell)$ consisting of the following variables and  (marginal) events. 
        \begin{itemize}
          \item \textbf{Bad Events:} 
                      $\cB_{\post} = \{ (\cE | \psi_{pre}) : \cE\in \cB, \Pr(\cE \mid \psi_{\pre}) > 0\}$
        \item \textbf{Variables:} $\cV_{\post}=\cV_2\cap \cup_{\cE\in \cB_{\post}}\vbl(\cE)$, with their respective original probability distribution.  The variables in $\cV_2 \setminus \cV_{\post}$ are not part of the postshattering phase; their values are set arbitrarily. 
        \item $\ell_{\post}(x)=\ell(x)$ for all $x\in \cV_{\post}$ and $\ell_{\post}(\cE|\psi_{\pre})=\ell(\cE)$ for all $(\cE|\psi_{\pre})\in \cB_{\post}$. To simplify the notation, we refer to $\ell_{\post}$ as $\ell$.
        \end{itemize}
        
        We compute an assignment $\psi_{\post}$ of the variables in $\cV_{\post}$ that avoids all events in $\cB_{\post}$. In \LOCAL this is done via \Cref{thm:deterministicLLLLOCAL} and in \CONGEST we use the algorithm from  \Cref{lem:CONGESTpostshattering}.
        \item \textbf{Return} $\psi=\psi_{\pre}  \cup \psi_{\post}$~.
\end{itemize}

\myparagraph{Analysis of the post-shattering phase.}
We first show that the LLL formulation in the post-shattering phase indeed is an LLL. Afterwards, we show that with high probability the dependency graph of this LLL consists of small connected components. In the \CONGEST model, we actually require a stronger statement, that is, we require that the projection of the post-shattering LLLs to the communication graph only consists of small components.  To simplify the notation in the proofs,  we identify the events $\cE'=(\cE|\psi_{\pre})$ in $\cB_{\post}$ with their corresponding event $\cE\in \cB$.

\begin{lemma}
\label{lem:twoVariableLLLpostShattering}
 $\cL_{\post}$ is an LLL with bad event probability bound $p$ and dependency degree $d$. 
\end{lemma}
\begin{proof}
The dependency degree of $\cL$ is upper bounded by the $d$ in \Cref{def:twoVariableLLL}, and thus so is the dependency degree of $\cL_{post}$. 
By \Cref{def:twoVariableLLL}, we obtain $\Pr(\cE')=\Pr(\cE\mid \phi_\pre)\leq p$ for any event $\cE'=(\cE\mid \phi_\pre)$  in $\cB_{\post}$.
\end{proof}

Note that the statement in \Cref{lem:twoVariableLLLpostShattering} holds ``deterministically'', that is, regardless of the outcome of the random choices in Step~1 of the algorithm. 

Let $W=\{v\in V | \ell^{-1}(v)\cap(\cV_{\post}\cup \cB_{\post})\neq\emptyset\}$ be the set of nodes that have one of their events/variables participating in the post-shattering phase. Further, let  $W'=\{v\in V : \dist_G(v,W)\leq \nu\}$ be the set of nodes within 
distance $\nu$ from $W$. 

We obtain the following bounds for events, variables, and nodes to be part of the post-shattering. 
\begin{lemma} The following bounds hold. 
\label{lem:TwoVariableshatteringProbabilities}
\begin{itemize}
    \item For each event $\cE\in \cB$ we have $\Pr(\cE\in \cB_{\post})\leq p$~,
    \item For each variable $x\in \cV$ we have $\Pr(x\in \cV_{\post})\leq  p\cdot (d+1)$~,
    \item For each node $v\in V(G)$ we have $\Pr(v\in W)\leq p \cdot  (d+1)\cdot l$~, 
    \item For each node $v\in V(G)$ we have $Pr(v\in W')\leq p \cdot   (d+1)\cdot l\cdot \Delta^{\nu}$.
\end{itemize}
For distinct $\cE$ and $\cE'$ the events $\cE\in \cB_{\post}$ and $\cE'\in \cB_{\post}$ are independent if $\dist_H(\cE,\cE')> 1$. For a node $v$ the event whether it is contained in $W$ or $W'$ only depends on the randomness at nodes $v'$ with $\dist_G(v,v')\leq 2\nu$ and $\dist_G(v,v')\leq 3\nu$, respectively. 
\end{lemma}
\begin{proof}
Fix an arbitrary event $\cE\in \cB$. By \Cref{def:twoVariableLLL}, the probability that $\cE$ is contained in $\cB_{\post}$ is at most $p$. 
Fix an arbitrary variable $x\in \cV$. In order for $x\in \cV_{\post}$ to hold, one of its $d_{\cV}$ adjacent events needs to participate in $\cB_{\post}$. We obtain $\Pr(x\in \cV_{\post})\leq  p \cdot d_{\cV}$ with a union bound over these events.

Fix an arbitrary node $v\in V$. For $v\in W$ to hold, one of its $l(v)$ events/variables needs to be contained in $\cB_{\post}\cup \cV_{\post}$. We obtain $\Pr(v\in W)\leq p\cdot l(v)  \cdot d_{\cV}\leq p\cdot l \cdot d_{\cV}$ with a union bound over all these events/variables. Similarly, $v\in W'$ only holds if there is a node $u\in N^{\nu}(v)$ with $u\in W$. We obtain $\Pr(v\in W')\leq p\cdot l \cdot d_{\cV}\cdot \Delta^{\nu}$ with a union bound over the $\Delta^{\nu}$ nodes in $N^{\nu}(v)$.

 The event $(\cE\in \cB_{\post})$ only depends on the random choices of the variables in $\vbl(\cE)\cap \cV_1$. Thus, $(\cE\in \cB_{\post})$ and $(\cE'\in \cB_{\post})$ for events $\cE,\cE'$ depend on distinct sets of variables if $\dist_H(\cE,\cE')>1$. 

For a node $v\in V$ the event $(v\in W)$ holds if one of the events/variables in $\ell^{-1}(v)$ participate in $\cV_{\post}\cup \cB_{\post}$. For a variable $x\in \ell^{-1}(v)$, this only depends on the outcome of all variables in $\vbl(\cE)$ for each event $\cE$ with $x\in \vbl(\cE)$. 
These variables are necessarily simulated by a node in distance at most $2\nu$. For an event $\cE\in \ell^{-1}(v)$, whether the event $\cE\in \cB_{\post}$ holds only depends on the variables in $\vbl(\cE)\cap \cV_1$, which are contained in $N^{\nu}(v)$.  For the last claim, this distance increases by an additive $\nu$ by the definition of $W'=N^{\nu}(W)$.

To obtain the bounds in the lemma we use that $d\geq d_{\cV}-1$ as all events incident to a variable are dependent. 
\end{proof}
For intuition regarding the next lemma, we refer to the text before \Cref{lem:sampleLLLshattering}.
Observe that two events in the dependency graph $H_{\cL_{\post}}$ of $\cL_{\post}$ are connected by an edge if they share a variable in $\cV_{\post}$. So, $H_{\cL_{\post}}\subseteq H[\cB_{\post}]$ holds.

\begin{lemma}[shattering]
\label{lem:twoVariableComponents}
If $p\leq d^{-14}$ holds, then with high probability each connected component of $H[\cB_{\post}]$ contains at most $O(d^{4}\cdot \log n )$ nodes.

If $p < ((d+1)\cdot l)^{-1}\Delta^{-(4+12\nu)}$ holds, then with high probability the graph $G[W']$ consists of connected components of size $O(\log n \cdot \Delta^{6\nu})$. \end{lemma}
\begin{proof}
By \Cref{lem:TwoVariableshatteringProbabilities} and the shattering \Cref{lem:Shattering} (with $c_3=2, c_2=1, c_1=9$), we obtain that with high probability in $n$ the connected components of $H[\cB_{\post}]$ are of size at most $O(\log n\cdot d^2)$ if $p\leq d^{-14}$ holds. This with high probability guarantee depends on the randomness of Step~1.

By \Cref{lem:TwoVariableshatteringProbabilities} and the shattering \Cref{lem:Shattering} (with $c_3=2, c_2=3\nu, c_1=4+12\nu$), we obtain that with high probability in $n$ the connected components of $G[W']$ are of size at most $O(\log n\Delta^{6\nu})$ if $\Pr(v\in W')\leq p ((d+1)\cdot l\cdot\Delta^{\nu})\leq  \Delta^{-c_1}$ holds. This high probability guarantee depends on the randomness of Step~1. 
\end{proof}
The proof of the following lemma is identical to the proof of \Cref{lem:promiseSmallComponentContained}.
\begin{lemma}
\label{lem:twoSetsmallComponentContained}
The projection of each connected component of $H[\cB_{\cL_{\post}}]$ via $\ell$, including all incident variables in $\cV_{\cL}$ of the events in $\cB_{\cL_{\post}}$, is contained in a single connected component of $G[W']$.
\end{lemma}

\begin{proof}[Proof of \Cref{thm:LLLTwoSets}]
First note that Step~1 can be executed in $\poly\log\log n$ rounds in \LOCAL if the event/variable assignment $(\cL,G,\ell)$ has constant locality. In the \CONGEST model, it can be done in $\poly\log\log n$ rounds if the LLL is simulatable. Note, that each event $\cE=\cE_1\cap\cE_2$ can also test in $\poly\log\log n$ rounds whether $\cE_1$ is avoided after Step~1 as for disjoint variable set LLLs this is required from the simulatability definition (see the line after \Cref{def:twoVariableLLL}).

At the end of the algorithm each event $\cE=\cE_1\cap \cE_2\in \cB$ is avoided under $\phi$ for as either $\cE_1$ is avoided under $\psi_{\pre}$ which agrees with $\phi$, or $\cE_2$ is avoided under $\psi_{\post}$ which also agrees with $\phi$.

\Cref{lem:twoVariableLLLpostShattering} shows that $\cL_{\post}$ indeed is an LLL with dependency degree $d$ and upper bound $p$ on the bad event probability. 

In the \LOCAL model we solve $\cL_{\post}$ as follows. \Cref{lem:twoVariableComponents} shows that the connected components of $H[\cB_{\post}]$ are of size $O(d^{4} \log n )$ if $p<d^{-14}$. As the dependency graph of $H_{\cL_{\post}}$ is a subgraph of $H[\cB_{\post}]$ its connected components are of the same size. Now, we can independently apply \Cref{thm:deterministicLLLLOCAL} to solve each of these instances in parallel in $\poly\log N=\poly\log\log n$, where each round of communication in $H_{\cL_{\post}}$ can be simulated in $O(\nu)$ rounds in the communication network~$G$.

In the \CONGEST model, we solve $\cL_{\post}$ as follows. Formally, set up a new LLL $\cL'_{\post}$ in which we append $\cL_{\post}$ by the variables $\cV_{\post,\passive}=\bigcup_{\cE\in \cB_{\post}}\vbl(\cE) \setminus \cV_{\post}$. Also, we set $\ell(x)=\ell(x)$ for all $x\in \cV_{\post,\passive}$. Note that all these variables already have an assignment in $\psi_{\pre}$. Due to $\Delta\leq \log^c n$, $l\leq d^{c_l}$, $d_{\cV}\leq d$ and the condition on $p$ in \Cref{thm:LLLTwoSets}, we obtain
\begin{align}
\label{eqn:shatteringCondition}
    p<d^{-(2+c_l)-(4c+12c\nu)\log\log n}\leq  ((d+1)\cdot l)^{-1}\Delta^{-(4+12\nu)}~.
\end{align} 
In other words, the conditions of \Cref{lem:twoVariableComponents} are met for sufficiently large $n$. Thus, the lemma shows that the connected components of $G[W']$ are of size $N=O(\log n\Delta^{6\nu})=O(\log^{6c\nu+1} n$) and the same holds for each connected component of $H[\cB_{\post}]$. By \Cref{lem:twoSetsmallComponentContained}, we also have that the projection via $\ell$ of each connected component of $H[\cB_{\post}]$ is contained in a connected component of $G[W']$ (this includes the assignment of $\cV_{\post,\passive}$ to nodes in $G$).  Hence, we can apply \Cref{lem:CONGESTpostshattering} in parallel on all components of $\cL'_{\post}$ with the partial assignment $\psi_{\pre}$ restricted to variables incident to events in $\cB_{\post}$. This solves $\cL'_{\post}$ in $\poly\log N=\poly\log\log n$ rounds and be correct with probability at least $1-2^{\bandwidth}=1-n^{-2}$, if $\bandwidth=2\log n$. In the application of \Cref{lem:CONGESTpostshattering}, we set $\lambda=10\log\log n$, which satisfies the condition on the LLL criterion in \Cref{lem:CONGESTpostshattering} as $p<d^{-(2+c_l)-(4c+12c\nu)\log\log n}\leq p^{-\lambda}$. As $\lambda=\Omega(\log N)$, the algorithm runs in $\poly\log\log n$ rounds by \Cref{cor:CONGESTpostshattering}. This also solves $\cL_{\post}$. 
\end{proof}

\section{Efficient Post-shattering in CONGEST}
\label{sec:CONGESTpostshattering}
In this section, we devise \CONGEST algorithms for the LLL subinstances to be solved after shattering certain LLLs.

Known solutions for LLL in the \CONGEST model cannot make use of the small component size unless $d$ and $\Delta$ are at most $\poly\log\log n$ \cite{MU21,HMN22}. Thus, we develop a novel algorithm to efficiently solve LLLs on small components in $\poly\log\log n$ rounds. 
Recall that, intuitively, an LLL instance is simulatable
if one can (a) evaluate the status of each event in $\poly\log\log n$ rounds (note the variables of an event might be simulated by a node that is a few hops from the node that simulates the event), and (b) nodes can determine (in $\poly\log\log n$ rounds) the influence on their simulated bad events by partial assignments of variables. There are additional conditions to the definition of simulatability taking into account some allowed pre-processing and specifying the bandwidth that is allowed for the respective steps. 

The goal of this section is to prove the following lemma that we need in the post-shattering phases of our algorithms. 

\begin{restatable}{lemma}{CONGESTpostshattering}
\label{lem:CONGESTpostshattering} Let $\lambda>0$ be a (possibly non-constant) parameter.
There is a $\CONGEST(\mathsf{bandwidth})$ algorithm  for  LLL instances $\cL=(\cV,\cB,\ell)$ with dependency degree is $d$ and for which $\Pr(\cE\mid \psi) <d^{-\lambda}$ holds for the marginal error probability of all events $\cE\in \cB$ where $\psi$ can be any partial assignment of the variables in $\cV$.

The algorithm requires that the locality $\nu$ of the LLL instances is constant and that each connected component of $G[N^{\nu}(\ell(\cV\cup \cB))]$ is of size at most $N$. It works with an ID space that is exponential in $N$. It requires that the LLL instance is simulatable and runs in $\poly\log\log n\cdot \log (N\cdot l) + N^{2/\lambda}\cdot \poly\log N$ rounds where $l$ is the load of the LLL.  The error probability is upper bounded by $2^{-\mathsf{bandwidth}}$. 
\end{restatable}

\begin{corollary}
\label{cor:CONGESTpostshattering}
If $\lambda =\Omega(\log N)$, $x\leq \poly\log\log n$, $l\leq \poly\log\log n$, and $\bandwidth=\Omega(\log n)$, the instances as in \Cref{lem:CONGESTpostshattering} can be solved in $\poly\log\log n$ rounds with high probability in $n$.
\end{corollary}

\subsection{Network Decomposition}
\label{sec:networkDecomp}
A \emph{weak distance-$k$ $(C,\beta)$-network decomposition with congestion $\kappa$} is a partition of the vertex set of a graph into clusters $\cC_1,\ldots,\cC_p$ of (weak) diameter $\leq \beta$, together with a color from $[C]$ assigned to each cluster such that clusters with the same color are further than $k$ hops apart. Additionally,  each cluster has a communication backbone, a Steiner tree of radius $\leq \beta$, and each edge of $G$ is used in at most $\kappa$ backbones. For additional information on such decompositions, we refer the reader to \cite{MU21,GGR20}. For the sake of our proofs, we only require that such decompositions can be computed efficiently (\Cref{thm:networkDecompMU21}) and that one can efficiently aggregate information in all clusters of the same color in parallel in time that is essentially proportional to the diameter $\beta$ (see \Cref{cor:treeAggregationBetter} for the precise statement). 
 \begin{theorem}[\cite{MU21}]
 	\label{thm:networkDecompMU21}
  For any (possibly non constant) $k\geq 1$ and any  $\lambda\leq \log N$ there is a deterministic \CONGEST algorithm that, given a graph $G$ with at most $N$ nodes and unique IDs from an exponential ID space, computes  a weak $(\lambda, k\cdot N^{1/\lambda} \log^3 n)$-network decomposition of $G^k$ with congestion $\kappa=O(\log N \cdot \min\{k,N^{1/\lambda}\cdot \log^2 N \})$ in $O(k\cdot N^{2/\lambda}\cdot \lambda\cdot \log^6 N \cdot \min\{k,N^{1/\lambda}\log^2 N\})$ rounds.  
   \end{theorem}
While we state our results in terms of general $\lambda$, the most natural invocation of \Cref{thm:networkDecompMU21} is $\lambda=\log N$ and constant $k$, in which case it computes a weak $(\log N, \log^3 N)$-network decomposition of $G^k$ with congestion $O(\log N)$ in $O(\log^7 N)$ rounds. 
 
There are algorithms to compute similar network decompositions that are more efficient than the one in \Cref{thm:networkDecompMU21}, e.g., in \cite{GGH23, MPU23}. However, as stated in their results, they require $\Omega(\log N\cdot \log\log N)=\omega(\log N)$ colors, and as the number of colors factors into the LLL criterion of the instances we can solve and for the sake of simplicity, we refrain of using their algorithms in a black box manner. 

\subsection{The LLL algorithm of Chung, Pettie, Su }

In the \LOCAL model there is a simple $O(\log n)$-round randomized distributed algorithm for LLLs with $epd^2<1$ and constant locality \cite{CPS17}. 
\begin{algorithm} \caption{The simple LLL algorithm from \cite{CPS17} (Theorem~\ref{thm:CPS}). 
} \label{alg:CPSLLL}
\begin{flushleft}
    Initialize a random assignment of the variables \\
	Let $\mathcal{F}$ be the set of bad events under the initial assignment\\
	\textbf{While }{$\mathcal{F} \neq  \emptyset$}{~\\
        \hspace*{1.5em} Let $I = \big\{A \in\mathcal{F} : ID(A) = \min\{ID(B) | B \in  N_{\mathcal{\cH_{\cL}}}(A)\}\big\}$ \\
	  \hspace*{1.5em} Resample $\vbl(I) = \cup_{A\in I}\vbl(A)$  \\
  	\hspace*{1.5em} Let $\mathcal{F}$ be the set of bad events under the current assignment
	}
\end{flushleft}
\end{algorithm}

It is important for our paper that \Cref{alg:CPSLLL} can be implemented in the \CONGEST model if natural primitives are available. These primitives are the evaluation of events and the computation of local ID-minima (in the dependency graph) of a subset of the events for an arbitrary ID assignment (that may even be adversarial). The algorithm is used as a subroutine in our post-shattering solution. 

\begin{theorem}[\cite{CPS17}]
\label{thm:CPS}
Suppose \Cref{alg:CPSLLL} is run for 
$O(\log_{\nicefrac{1}{epd^2}} |\cB|)$ iterations on an LLL $\cL=(\cV,\cB)$ 
satisfying $epd^2<1$.
Then, with probability at least $1-1/|\cB|$, it computes an assignment 
that avoids all bad events. 
The algorithm requires events to be equipped with identifiers that are unique within their connected component of the dependency graph.
\end{theorem}

\subsection{Efficient Post-shattering in \CONGEST (details)}
\label{sec:congestPostShatteringProof}
To devise an efficient \CONGEST post-shattering algorithm, we decompose each small component into small clusters via the network decomposition algorithm from \Cref{thm:networkDecompMU21}. Then, the objective is to iterate through the collections of the decomposition and when processing a cluster we want to fix all variables in that cluster. When doing so we need to ensure that the influence on (neighboring) clusters processed in later stages of the algorithm is not (too) negative. To efficiently measure and limit this effect to a $d^2$-factor increase in the marginal probability of each remaining event, we set up a new LLL for each cluster that ensures just that condition. An application of Markov's inequality shows that the probability of the marginal probability to increase more than a factor $d^2$ if a subset of the variables of an event is sampled is at most $1/d^2$ (see \Cref{claim:LLLMarkov}). 
Implementing all required primitives efficiently needs the simulatability of the LLL. 

In \Cref{sec:twoVariableLLL,sec:samplingLLL}, we want to use \Cref{lem:CONGESTpostshattering} to exploit that the post-shattering phase consists of several small connected components.

\begin{proof}[Proof of \Cref{lem:CONGESTpostshattering}]
 First, we use \Cref{thm:networkDecompMU21} to compute a network decomposition of the communication network $G$ where the distance between two clusters of the same color is strictly larger than $2\nu$ (recall, that $\nu=O(1)$ is the locality of event/variable assignment $\ell$). We instantiate the theorem such that we obtain $\lambda$ collections, weak diameter $O(N^{1/\lambda} \log^3 N)$ and congestion $O(\log N)$. The algorithm runs in $O(\lambda\cdot N^{2/\lambda}\log^6N)$ rounds as $\nu$ is constant. This is obtained by running the algorithm on all connected components of $G[N^{\nu}(\ell(\cV\cup \cB))]$ in parallel. 

To compute an assignment of $\cV$ that avoids all events in $\cB$, we start with an initial partial assignment $\phi_0$ with all variables unassigned 
and then iterate through the $\lambda$ collections of the network decomposition. 
When processing the nodes $V_i$ in collection $i$, we permanently assign the variables in $\cV_i=\cV\cap \ell^{-1}(V_i)$. Let $\phi_i$ be the partial assignment after processing the $i$-th collection. The invariant that we maintain for each event $\cE\in \cB$ is that after processing the $i$-th collection we have $\Pr(\cE\mid \phi_i)\leq p\cdot d^{2i}$. Initially, the invariant holds for $\phi_0=\psi$ ($\psi$ from the lemma statement) by the upper bound $p$ on bad event probabilities of events in $\cB$. Let $\phi=\phi_\lambda$ be the final assignment in which each variable of $\cV$ has a value $\neq\bot$. 
We obtain that $\Pr(\cE\mid \phi)\leq p d^{2\lambda} < 1$, since $p < d^{-2\lambda}$.
Since all variables have been fixed by $\phi$, 
this means that the event $\cE$ is avoided under $\phi$.

Next, we detail the process for a single collection, that is, how to find the partial assignment $\phi_i$, given $\phi_{i-1}$. Note that throughout the algorithm all partial assignments are known in a distributed manner, that is, for a variable $x\in \cV$, the node $\ell(x)$ knows the value of $x$ (or $\bot$) in the current partial assignment. Also observe that for each variable $x\in \cV$ node $\ell(x)$ knows whether $x\in \cV_i$ holds. Let $\cB_i$ be the set of events $\cE$ with $\vbl(\cE)\cap \cV_i\neq\emptyset$. By using the aggregation primitive (here, we only require a broadcast) from the variables in $\cV_i$ to the events in $\cB_i$  nodes learn whether their events are contained in $\cB_i$. We define the following LLL $\cL_i$ for collection $i$.

\begin{itemize}
\item \textbf{$\cV_{\cL_{i}}$ set of variables:} $\cV_{i}$ with their original distribution.
\item \textbf{$\cB_{\cL_{i}}$ set of bad events:} For each $\cE\in\cB_{i}$ there is an event $\cE'$ that holds on an assignment $\psi$ of $\cV_{i}$ if 
\begin{align*}
\Pr(\cE\mid \psi\cup \phi_{i-1})\geq d^2\Pr(\cE\mid \phi_{i-1}), 
\end{align*}
\item $\ell_{\cL_{i}}(x)=\ell(x)$ for all $x\in \cV_{\cL_{i}}$ and $\ell_{\cL_{i}}(\cE')=\ell(\cE)$ for all $\cE'\in \cB_{\cL_{i}}$. 
\end{itemize}
\begin{claim}
\label{claim:LLLMarkov}
$\cL_{i}$ is an LLL with error probability at most $1/d^2$ and dependency degree $d$.
\end{claim}
\begin{proof}
Consider an event $\cE'\in \cB_{\cL_{i}}$ and let $\cE\in \cB_{\cL}$ be the corresponding event of $\cL$. 
For an assignment $\psi$ of $\cV_{\cL_{i}}$, let $p_{\psi}=\Pr(\cE\mid \psi\cup \phi_{i-1})$. Note that formally $p_{\psi}$ is a random variable over the randomness of the variables in $\cV_{\cL_{i}}$. For each assignment of these variables $p_{\psi}$ is a value in $[0,1]$. Hence, we obtain $E_{\cV_{\cL_{i}}}[p_{\psi}]=\Pr(\cE\mid \psi\cup \phi_{i-1})\eqqcolon\mu$, where the subscript indicates that the randomness of the expectation is only for the variables in ${\cV_{\cL_{i}}}$. By Markov inequality, we have $\Pr(\cE')=\Pr(p_{\psi}\geq d^2\mu)\leq 1/d^2$. 
\end{proof}
We use the following claim (sometimes implicitly) throughout the proof. 
\begin{claim}
\label{claim:componentContained}
Each connected component of the dependency graph of $H_{\cL_i}$ is contained in a connected component of $G[N^{\nu}(\ell(\cV\cup \cB))]$.
\end{claim}

Now, we run $k=\bandwidth$ parallel instances of the LLL algorithm of \cite{CPS17} with the LLL $\cL_{i}$. Each instance runs for $O(\log (N\cdot l))$ iterations (note that each connected component of $G^\nu[\ell(\cV\cup \cB)]$ contains at most $N\cdot l$ events); the precise discussion of the total runtime of executing these instances and the remaining parts of the algorithm is deferred to the end of the proof. As a result we obtain $k$ assignments $\psi_1,\ldots,\psi_k$ of $\cV_{i}$.  For $j\in [k]$, we say that an assignment $\psi_j$ is \emph{correct} for an event $\cE'\in \cB_{\cL_{i}}$ if $\cE'$ is avoided under $\phi_j$. 

To define $\phi_i$ from the assignments $\psi_1,\ldots,\psi_k$, let us define the following notation for each cluster $\cC$ with color $i$ in the network decomposition. Let $\cV_{\mathcal{C}}=\cV_i\cap \ell^{-1}(\cC)$ and $\cB_{\cC}=\{\cE\in \cB_i: \vbl(\cE)\cap \cV_{\cC}\neq\emptyset\}$. For two distinct clusters $\cC$ and $\cC'$ of the $i$-th collection we obtain that $\cV_{\cC}\cap \cV_{\cC'}=\emptyset$ and $\cB_{\cC}\cap \cB_{\cC'}=\emptyset$ as the cluster separation is strictly larger than $2\nu$. By the properties of \Cref{alg:CPSLLL} (and using \Cref{claim:componentContained}), each $\psi_j$  is correct for all events in $\cB_{\cC}$ with probability $\geq (1-1/(N\cdot l))\geq 1/2$. Thus, with probability at least $1-1/2^k$ there is an $j_{\cC}^*\in [k]$ such that $\psi_{j_{\cC}^*}$ is \emph{correct} for all events of $\cB_{\cC}$.  Each node holding an event of $\cL_{i}$ determines which assignments are correct for its event. For each cluster $\cC$ in parallel the nodes in $\ell_{\cL_i}(\cB_{\cC})$ agree on an index $j_{\cC}^*$ such that assignment $\psi_{j_{\cC}^*}$ is correct for all events in $\cB_{\cC}$; let $\tilde{\psi}_{j_{\cC}^*}$ denote the restriction of this assignment to $\cV_{\cC}$. Different clusters may decide on different indices. Lastly, we set $\phi_i=\phi_{i-1}\cup\bigcup_{\cC\text{ has color $i$}}\tilde{\psi}_{j_{\cC}^*}$. The partial assignment $\psi_i$ is well-defined as each variable of $\cV_i$ appears in exactly one cluster.

\begin{claim}
     The invariant $\Pr(\cE\mid \phi_i)\leq p\cdot d^{2i}$ holds. 
\end{claim}
\begin{proof}
For each event $\cE\notin \cB_i$, the claim follows as none of the variables in $\vbl(\cE)$ change their value when processing the $i$-th collection, that is, we obtain $\Pr(\cE\mid \phi_i)=\Pr(\cE\mid \phi_{i-1})\leq p\cdot d^{2(i-1)}\leq p\cdot d^{2i}$. 

For each event $\cE\in \cB_i$, there is a cluster $\cC$ with color $i$ for which $\cE \in \cB_{\cC}$. The partial assignment $\tilde{\psi}_{j_{\cC}^*}$ avoids all events in $\cB_{\cC}$.
As $\Pr(\cE\mid \phi_{i-1})\leq p \cdot d^{2(i-1)}$ holds, the definition of the avoided events in $\cB_{\cC}\subseteq \cB_{i}$  implies
    $\Pr(\cE\mid \psi\cup \phi_{i-1})\geq d^2\Pr(\cE\mid \phi_{i-1})\leq p d^{2i}$~.
\end{proof}

Before we run the $k$ parallel instances of \Cref{alg:CPSLLL}, we compute a locally unique ID $\iota(v)\in \{1,\ldots,N\}$ for each node $v\in 
\ell_{\cL_i}(\cB_i)$, for each cluster $\cC$ of color $i$. Then, for each event $\cE\in \cB_{i}$ assign a locally unique ID $\iota(\cE)=(\iota(\ell_{\cL_i}(\cE),j)\in \{1,\ldots,N\}^2$ where $j$ is the index of $\cE$ among the events simulated by node $\iota(\ell_{\cL_i}(\cE)$ (according to an arbitrary ordering of these events). These IDs are used to compute the local minima of non-avoided events in the simulation of \Cref{alg:CPSLLL}. As the distance between clusters of the same color is strictly larger than $2\nu$, these IDs are unique within the connected components of the dependency graph of $H_{\cL_i}$, which tailors them sufficiently for the simulation of \Cref{alg:CPSLLL}. Assigning these locally unique IDs can be done in $O(N^{1/\lambda}\poly\log N)$ rounds by using the Steiner tree of the clusters. 

We next, bound the runtime of the whole process and show that all steps can be executed. To improve the readability let $x=\poly\log\log n$ be the number of rounds from \Cref{def:simulatability}. Recall, that nodes know whether their simulated variables and events are contained in $\cV_i$ and $\cB_i$, respectively. As the LLL $\cL$ is simulatable, we only need to send $k\cdot x$ bits to run one round of each of the $k$ instances of \cite{CPS17}. In more detail, as the new IDs that we computed are represented with bit strings of length $O(\log N)$, the pre-processing phase of \Cref{def:simulatability} requires $x\cdot  \mathsf{IDbitLength}=x\cdot \poly\log N$ rounds of $\CONGEST(\bandwidth)$. Then \Cref{alg:CPSLLL} applied to $\cL_{i}$ can be implemented with the operations in the second part of \Cref{def:simulatability}. More detailed, in $x$ rounds, for each event $\cE$, the corresponding nodes can determine whether $\Pr(\cE\mid \psi\cup \phi_{i-1})\geq d^2\Pr(\cE\mid \phi_{i-1})$ holds or not, where $\psi$ is an assignment of the variables in $\cV_i$. Also, we can compute a set of local minimum IDs $I$ (in the dependency graph $H_{\cL_i}$) of events $\cE$  for which this is not the case by using the broadcast and aggregation primitives.  Let $Z$ be the set of non-avoided events in some iteration. Each event $\cE\in Z$ broadcasts its locally unique ID $\iota(\cE)$ to all variables in $\vbl(\cE)\cap \cV_i$.  For a variable $x\in \cV_i$ let $c(x)=\min_{\cE\in Z}\iota(\cE)$. Then the variables use the aggregation primitive to send the smallest ID that they have received to the events that contain them, that is, each event $\cE$ receives $c(\cE)=\min_{x\in \vbl(\cE)\cap \cV_i} c(x)$. The set $I=\{\cE : c(\cE)=\iota(\cE)\}$ is consists of locally minimum IDs of $H_{\cL_i}[Z]$, as required. Lastly, events in $I$ use the broadcasting function to inform their respective variables such that they can re-sample themselves for the next iteration of \Cref{alg:CPSLLL}.

Thus, simulating one round of all $k$ instances of \Cref{alg:CPSLLL} in parallel requires $O(k\cdot x/\textsf{bandwidth})=x$ rounds as $\bandwidth=k$. Hence, running all instances for $O(\log N)$ rounds takes only $x\cdot \log (N\cdot l)$ rounds. The simulatability also implies that nodes in $\ell_{\cL_{i}}(\cB_{\cL_{i}})$ can check which events hold in which of the $k$ assignments $\psi_1,\ldots,\psi_k$  in $x$ rounds. 

Agreeing on the index $j_{\cC}^*$ in cluster $\cC$ can be done efficiently as follows: Each node $v\in \ell_{\cL_i}(\cB_{\cC})$ holds a bit string of length $k=\bandwidth$ in which the $j$-th bit equals $1$ if and only if all its events, i.e., the events in $\cB_{i}\cap\ell_{\cL_i}^{-1}(v)$, are avoided in $\psi_{j}$.  All nodes in $\ell_{\cL_i}(\cB_{\cC})$ agree on the index $j_{\cC}^*$ in time linear in the cluster's weak diameter (ignoring congestion between different clusters for now) by computing a bitwise-\textsf{AND} of the bitstrings. 

When simulating the $k$ instances of \Cref{alg:CPSLLL}, there is no congestion between clusters that are processed simultaneously as their distance in the graph is strictly greater than $\nu$. Agreeing on the assignment of the variables within one cluster is done on the Steiner tree of the cluster, where each edge is contained in at most one tree of clusters of the same color. Since also the clusters' weak diameter is of size $\poly\log N$ and as we have $O(\log N)$ colors classes the whole process runs in $x\cdot \log (N\cdot l)+\poly\log N$ rounds, and the claim follows as $x=\poly\log\log n$.
\end{proof}

\section{Applications and Bounding Risks}
\label{sec:applications}
 \Cref{sec:exampleSlackGeneration} serves as a warm-up. We sketch how to use our disjoint variable set LLL in order to solve the slack generation problem. As this assumes that we have already solved the DSS problem, the formal statement and proof appear are deferred to \Cref{sec:coloringSparseAndSlackGeneration}. 
 In \Cref{sec:bounding-risk}, we present several techniques to bound the risk events and show that our framework solve all LLLs that can be solved with the main LLL algorithm of \cite{GHK18}.
 In \Cref{sec:MonotoneLLLsExamples} we bound the risk of several types of binary LLLs and show that all LLLs that can be solved with the main LLL algorithm of \cite{GHK18} can be solved with out framework (in the \LOCAL model). 

\subsection{Example of Disjoint Variable Set LLL: Slack Generation}
\label{sec:exampleSlackGeneration}
Let us illustrate the usefulness of disjoint variable set LLLs by having a closer look at one example that is important for our coloring results. 
Recall that the goal is to color some of the nodes of the graph such that each node $v$ receives some \emph{slack}, that is, (at least) two of its neighbors $w,w'\in N(v)$ are colored with the same color. 
A simple randomized algorithm to generate slack for such nodes is to activate each node with a constant probability and then activated nodes select a random candidate color. Then, the candidate colors are exchanged with their neighbors and a node gets permanently colored with its candidate color if no neighbor tried the same color. 

\begin{algorithm}[ht]\caption{\slackgeneration} \label{alg:slackgen}
\begin{flushleft}
    \textbf{Input:} $S \subseteq V$
	\begin{algorithmic}[1] 
		\STATE Each node in $v \in S$ is active w.p. $1/20$
		\STATE Each active node $v$ samples a color $r_v$ u.a.r. from $[\chi]$. 
		\STATE $v$ keeps the color $r_v$ if no neighbor tried the same color.  
 	\end{algorithmic}
\end{flushleft}
\end{algorithm}

In \Cref{sec:coloringSparseAndSlackGeneration}, we prove the following lemma that bounds the probability that this process provides nodes with slack. Variants of this result have appeared numerous times \cite{molloy2013coloring, EPS15,HKMT21}
\begin{lemma}[Slack generation,  \Cref{lem:slackgen-custom} simplified]
    \label{lem:slackgen-custom-simple}
    Let $\Delta_s, \chi$ be positive integers with $\chi \ge \Delta_s/10$. Let $S \subset V$ be a subset of nodes.  
    Consider a node $v \in V$ with at least $\overline{m}$ non-edges in $G[N(v) \cap S]$. Suppose that all nodes in $S$, as well as $v$, have at most $\Delta_s$ neighbors in $S$. After running SlackGeneration on $S$ with color palette $[\chi]$, the slack of $v$ is increased by $\Omega(\overline{m}/\chi)$, with probability at least $1-\exp(-\Omega(\overline{m}/\chi))$. This holds independent of random choices at distance more than 2. 
\end{lemma}
\Cref{lem:slackgen-custom-simple} immediately gives rise to an LLL for $\chi=\Delta$. Introduce a bad event $\cE_v$ for each node $v$ of a graph that holds if $v$ does not obtain slack. Then, for sparse $\Delta$-regular graphs, that is, graphs in which any node has $\Omega(\Delta^2)$ non-edges in its neighborhood the probability that $\cE_v$ holds is upper bounded by $p=\exp(-\Omega(\Delta))$ and it only shares randomness with $d=\Delta^4$ events of other nodes. 

Consider a version of the slack generation problem where we are given 
two sets $S_1, S_2\subseteq V$, such that each node $v$ has many non-edges in the graph induced by $N(v)\cap S_1$ and also in the graph induced by $N(v)\cap S_2$. Observe that we can compute such sets by solving the DSS problem. 

Given these sets, we introduce the event $\cE_v=\cE_{v,1}\cap \cE_{v,2}$ for each node $v$ where for $i=1,2$ the event $\cE_{v,i}$ is avoided if $v$ receives slack from the coloring of the nodes in $S_i$. So, translating our algorithm for such LLLs into the language of this slack generation problem, we first execute SlackGeneration for all nodes in $S_1$. Then, each node $v$ that did not get slack yet, together with its neighbors in $S_2$ goes to the post-shattering phase. Here, we exploit that connected components are small and solve the slack generation problem for these nodes faster, e.g., by using \Cref{thm:deterministicLLLLOCAL} in the \LOCAL model or using \Cref{lem:CONGESTpostshattering} in the \CONGEST model.

\subsection{Techniques to Bound Risk}
\label{sec:bounding-risk}

The objective of this section is to illustrate techniques for bounding risk.
to show that all LLLs that can be solved with the main LLL algorithm of \cite{GHK18} can be solved with our framework. This section does not reason simulatability for an efficient CONGEST implementation, which is done elsewhere whenever these events are needed. 

Often we are just interested in the variables whose value is $\mathsf{black}$ and we sometimes abuse language and speak of these as \emph{sampled} variables/nodes; hence we also speak of a \emph{sampling LLL}. We illustrate in this section a number of natural sampling problems that have LLL with low \risk and can therefore be handled with our method. As we can modify the domain and the sampling, and we can combine criteria, this induces a space of solvable problems.

A property that appears frequently in sampling LLLs is that its bad events are monotone in the sense that they either profit from having more nodes set to \black or \white respectively. Examples are the events in the degree-bounded subgraph sampling and in the DSS problem discussed in the introduction.  
Formally, monotone events are captured by the following definition.

\begin{restatable}[monotone events]{definition}{defMonotoneEvents}
	An event $\cE$ defined over a set of independent binary random variables is \emph{$\mathsf{color}$-favoring} for $\mathsf{color}\in \{\black,\white\}$  if $\Pr(\cE \mid \psi)\leq \Pr(\cE \mid \phi)$ holds for any $\psi, \phi$ with $\psi(x) \leq \phi(x)$ for each $x \in \vbl(\cE)$ (where $\mathsf{color} < \bot < \overline{\mathsf{color}}$).
	\label{def:monotone}
\end{restatable}

We also call an event \emph{monotone increasing}  (\emph{monotone decreasing}) if it is $\black$-favoring (\white-favoring). The following lemma is one of the core advantages of our binary LLL solver, compared to using prior LLL algorithms (even when used in the LOCAL model).

\begin{restatable}[No Risk Lemma]{lemma}{lemRISKmonotone}
The \risk of a monotone increasing event $\cE$ is  $\Pr(\cE)$ testified by $\assoc(\cE)=\cE$~. 
\label{L:monotone-incr}
\end{restatable}

\begin{proof}
Let $\cE$ be a monotone-increasing event. Then $\cE$ itself testifies that the \risk of $\cE$ is at most $p$. First, we trivially have $\Pr(\cE)\le p$. We need to show that $\max_{\psi\in \mathsf{Respect}(\cE)}\{\Pr(\cE\mid \psi)\}\big\} \le p$. Let $\psi \in \mathsf{Respect}(\cE)$ be any promise retraction and let $\phi$ be a full assignment corresponding to $\psi$. By definition of $\mathsf{Respect}(\cE)$, either all white variables are retracted ($\phi(x) \in \{\black, \bot\})$, or no black variables are retracted. In the first case, $\psi(x) \le \bot$ for all $x \in \vbl(\cE)$, hence $\Pr(\cE \mid \psi) \le \Pr(\cE) \le p$. If no black variables are retracted, we have $\psi(x) \le \phi(x)$ for all $x \in \vbl(\cE)$. Hence, $\Pr(\cE \mid \psi) \le \Pr(\cE \mid \phi) = 0$, where the last equality follows from the fact that $\phi$ does not satisfy $\cE$ by definition of $\mathsf{Respect}(\cE)$. 
\end{proof}

The issue is that LLLs that only consist of monotone increasing events are trivial (one can simply set all variables to \black to solve them) and the asymmetry in the aforementioned approach (we only undo \white variables in the second step of retractions) prevents us from dealing with monotone decreasing events in the same manner simultaneously. Next, we present a general method for bounding the risk, mostly originating from prior work (and not used in our applications). 

\smallskip

\noindent\textbf{Bounding the risk of general events.} 
For completeness and to compare with prior work, we first explain how the risk of an event can be bounded by the analysis of LLL process given in prior work. The next few definitions and lemmas are reformulated in the language of this paper but appear in a similar manner in \cite{GHK18} and slightly differently already in \cite{CPS17}. Afterward, we discuss why these are helpful in general but unsuitable in the context of our efficient CONGEST algorithms.

Consider some event $\cE$ defined over some random variables $\vbl(\cE)$. Let $q$ be a parameter. A partial assignment $\phi$ of $\vbl(\cE)$ is \emph{$q$-dangerous} if there exists some retraction $\psi$ of $\phi$ such that 
\begin{align}
\label{eqn:conditional}
    \Pr(\cE\mid \psi)>q
\end{align} holds. Let $\cE^q_{\danger}$ be the event that holds on all assignments that are $q$-dangerous.
The following lemma follows immediately from the definition of risk and $\cE_{\danger}$. 
\begin{lemma}
\label{lem:danger} The risk of any event $\cE$ is upper bounded by $\max\{q,\Pr(\cE^q_{\danger})\}$ testified by $\cE^q_{\danger}$. 
\end{lemma}
As a result of \Cref{lem:danger} (applied with a suitable $q=1/\poly\Delta$), any LLL that can be solved with the main LLL algorithm of \cite[Section 6]{GHK18} can also solved with our LLL solver in the \LOCAL model.
Our algorithm is advantageous whenever one cannot easily bound $\Pr(\cE^q_{\danger})$, but can instead use \Cref{L:monotone-incr} to bound the risk.

In practice, it is highly non-trivial
to derive upper bounds on $\Pr(\cE^q_{\danger})$. 
 As an additional tool Ghaffari, Kuhn, and Harris introduce another technical term, the \emph{fragility} $f(\cE)$ of an event \cite[Definition 6.3]{GHK18}. 
\begin{definition}[Fragility]
    Let $\cE$ be an event on variables $\vbl(\cE) = \{x_1, \dots, x_k\}$ and let $\varphi_1, \varphi_2$ be two partial assignments with actual values for $\vbl(\cE)$ and $\bot$ for other variables. For any vector $a \in \{0,1\}^k$, define a new partial assignment $\psi_a$ by $\psi_a(x_i) = \varphi_{a_i}(x_i)$ for $1\le i \le k$ and $\psi_a(x_i) = \bot$ for $x_i \not\in \vbl(\cE)$. Let $\cE_B$ be the event 
    $$ 
        \cE_B = \bigvee_{a \in \{0,1\}^k} \cE \text{ occurs on assignment } \psi_a
    $$
    The fragility of $\cE$, denoted $f(\cE)$, is the probability of $\cE_B$ when $\varphi_1$ and $\varphi_2$ are drawn independently, according to the distribution of the variables in $\vbl(\cE)$. 
\end{definition}

 In \cite[Proposition 6.4]{GHK18} they also show that if an event $\cE$ has fragility $f(\cE)$, then the probability that  $\cE^q_{danger}$ holds is upper bounded by $f(\cE)/q$. 
\begin{lemma}
\label{lem:fragilityRisk}
The risk of an event $\cE$ is at most $\max(f(\cE)/q, q)$.
\end{lemma}

While these are powerful tools they come with two issues for our purposes: (1) the fragility and more general $\Pr(\cE^q_{\danger})$ can be quite hard to analyze for involved LLL processes with multiple dependencies such as DSS or the slack generation LLL, and (2) the associated event $\assoc(\cE)=\cE^q_{\danger}$ cannot, in general, be evaluated on an assignment in the \CONGEST model, as the respective conditional probabilities of \Cref{eqn:conditional} are unlikely to be computable without full information about all variables, the local graph structure in a graph problem, etc., none of which are readily available in the \CONGEST model.

The prime monotone decreasing events that appear in the context of our \CONGEST applications are events that control the maximum degree into some subgraph. In contrast to the involved lemmas above, it is easy to find a simple associated event. Alternatively, one could use a result from \cite[Theorem 6.8]{GHK18} to bound the fragility and hence via \Cref{lem:fragilityRisk} the risk.

\begin{lemma}\label{L:incr-LLL} 
Consider a random variable $X$ that is a sum of independent binary random variables. 
For some threshold parameter $x>0$, let $\cE_x$ be the event that $X>x$ holds. 
Then, the \risk of $\cE_x$ is at most $\Pr(\cE_{\nicefrac{x}{2}})$ testified by $\cE_{\nicefrac{x}{2}}$.

\end{lemma}
\begin{proof}
    By  definition, the risk testified by $\cE_{\nicefrac{x}{2}}$ is given by
    \begin{align*}
        \max\big\{\Pr(\cE_{\nicefrac{x}{2}}),\max_{\psi\in \mathsf{Respect}(\cE_{\nicefrac{x}{2}})}\{\Pr(\cE\mid \psi)\}\big\}~.
    \end{align*}
To prove the lemma, take an arbitrary assignment $\psi\in \mathsf{Respect}(\cE_{\nicefrac{x}{2}})$. As $\psi$ is a retraction of an assignment on which $\cE_{\nicefrac{x}{2}}$ was avoided at most $x/2$ of the variables in $\psi$ are \black. Let $\psi'$ be the assignment obtained from $\psi$ by setting all \black variables to \white. We obtain. 
\begin{align*}
\Pr(\cE_x \mid \psi) \leq \Pr(\cE_{\nicefrac{x}{2}} \mid \psi') \leq Pr(\cE_{\nicefrac{x}{2}})~, 
\end{align*}
which proves the claim.
\end{proof}

Note that monotone decreasing events, as in \Cref{L:incr-LLL}, do not rely on only \white variables being retracted in the second round of retractions (in fact, they prefer the \white variables to stay).  

The following lemma is useful to obtain upper bounds on the risk of combined events. 

\begin{lemma}
\label{L:retraction-cost-union}
Let $\cE_1$ and $\cE_2$ be events defined over the same set of independent random variables, and suppose they have \risk $p_1$ and $p_2$, respectively. Then the \risk of $\cE_1\cup \cE_2$ is at most $p_1+p_2$.

\end{lemma}
\begin{proof}
    Let $\cE'_1 = \mathsf{assoc}(\cE_1)$, $\cE'_2 = \mathsf{assoc}(\cE_2)$ be events testifying the \risk of $\cE_1$ and $\cE_2$. Let $\cE_{1,2}' = \cE'_1 \cup \cE'_2$. 
   We shall show that the event $\cE'_{1,2}$ testifies a \risk of at most $p_1 + p_2$ for $\cE_1 \cup \cE_2$. Firstly, $\Pr[\cE'_{1,2}] \le \Pr[\cE'_1] + \Pr[\cE'_2] \le p_1 + p_2$, by the union bound and the definition of \risk.
   We need to bound $\max_{\psi\in \mathsf{Respect}(\cE'_{1,2})}\Pr(\cE_1 \cup \cE_2\mid \psi)$. 
   Consider any $\psi\in \mathsf{Respect}(\cE'_{1,2})$. 
   Let $\hat{\phi}$ be a full assignment that respects $\psi$ and avoids $\cE'_{1,2}$, which exists by the definition of retraction. Since $\hat{\phi}$ avoids $\cE'_{1,2} = \cE'_1 \cup \cE'_2$, it in particular avoids $\cE'_1$ and $\cE'_2$.
    Hence, $\psi \in \mathsf{Retract}(\cE'_1)$ and $\psi \in \mathsf{Retract}(\cE'_2)$.
 Also, since $\psi$ satisfies the promises, 
$\psi \in \mathsf{Respect}(\cE'_1) \cap \mathsf{Respect}(\cE'_2)$.
Using this, we get $\Pr(\cE_1 \cup \cE_2\mid \psi) \le \Pr(\cE_1 \mid \psi) + \Pr(\cE_2 \mid \psi) \le p_1 + p_2$, by the assumption on the \risk of $\cE_1$ and $\cE_2$.
    Hence, $\cE'_{1,2}$ testifies the \risk of at most $p_1 + p_2$ for 
    the event $\cE_1 \cup \cE_2$.
\end{proof}

\subsection{Example LLLs with Low \Risk}
\label{sec:MonotoneLLLsExamples}
\label{sec:examples}
We illustrate here how we bound the risk of several types of natural binary LLLs.
In each of these problems, we assume that each node is sampled independently with probability $q$.

\begin{lemma}
\label{L:LLL-sample-bounds}
\label{lem:examplePromiseRetractionCost}
    Consider the following LLLs. 
    In all of them, each node is sampled independently with probability $q$, which induces a subgraph $S$. 
    Let $S_v = N(v) \cap S$ be the sampled neighbors of $v$.
    Let $\ell, \overline{\ell}$ be such that the induced neighborhood graph $G[N(v)]$ contains $\ell \cdot d(v)$ edges and thus $\overline{\ell} d(v) = (d(v)-\ell)d(v)/2$ non-edges.
We bound the risk of the following events. 
    \begin{enumerate}
        \item $\cE_v$: $|S_v| \ge 3q d(v)$, i.e.\ $v$ has more than $3 q d(v)$ sampled neighbors in $S$.
        \item $\cE_v$: $|S_v| \le q d(v)/2$, i.e.\ $v$ has fewer than $q d(v)/2$ sampled neighbors in $S$.
        \item $\cE_v$: $G[S_v]$ has fewer than $q^2 \overline{\ell}d(v)/2$ non-edges.  
        \item $\cE_v$: $G[S_v]$ has fewer than $q^2 \ell d(v)/2$ edges.
    \end{enumerate}
    Events (1) and (2) have  \risk at most $\exp(-\Omega(qd(v)))$. Event (3) has risk at most $\exp(-q \overline{\ell}/5)$ and (4) has risk at most $\exp(-q \ell/5)$. 
\end{lemma}
\begin{proof}
\begin{enumerate}
\item  Monotone decreasing event. We apply \cref{L:incr-LLL} to the variable $X_v = |S_v|$, with associated event $\cE'_v$ being that $|S_v| \ge 3qd(v)/2$. The \risk is at most $\Pr[\cE'_v]$, which by Chernoff is $\exp(-\Omega(qd(v)))$. 

  \item  Monotone increasing event for which we apply \cref{L:monotone-incr}. The \risk is at most $\Pr[\cE_v]$, which by Chernoff is $\exp(-\Omega(qd(v)))$. 
    
    \item Monotone increasing event. Represent the number of non-edges in 
    $G[S_v]$ by the random variable $X_v = |\{\{u,w\}\in S_v \times S_v : \{u,v\} \not\in E\}|$. The expected value of $X_v$ is $q^2 \overline{\ell} d(v)$. 
    By \cref{L:monotone-incr}, the \risk is at most $\Pr[\cE_v]$, which by \cref{lem:nonEdgeHittingV2} is at most $\exp(-q \overline{\ell}/5)$.

    \item Identical to part 3, switching the role of edges and non-edges.\qedhere
    \end{enumerate}
\end{proof}

The subsetting problems show that our method need not depend on the maximum degree, $\Delta$, but rather in terms of the size of the domain of the sampling.

These properties can be easily generalized to restricted forms of sampling. 

\begin{observation}
    These properties hold also when we sample from a subset $T \subseteq V$. The bounds are then in terms of $d_T(v) = |T \cap N(v)|$. 
    \label{obs:subsampling}
\end{observation}

We now use these to bound the risk of various subsampling problems. The lemma does not take simulatability of the (associated) events into account yet. 
\begin{lemma}
The following problems can be formulated as LLLs with bounded \risk:
\begin{enumerate}
  \item Vertex subset splitting: Split a given subset $T \subseteq V$ of nodes into two parts $V_1$, $V_2$, such that each node $v$ in $V$ has between $d_T(v)/2$ and $3 d_T(v)$ neighbors. \Risk: $\exp(-\Omega(d_T(v)))$.
  \item Sparsity splitting: We are given that each node has 
   at least $\ell d(v)$ non-edges within its neighborhood and wish to partition the vertices into two sets $S_1$, $S_2$, such that each node has 
   at least $\ell d(v)/16$ non-edges within $N(v) \cap S_1$ and within $N(v) \cap S_2$. \Risk: $\exp(-\Omega(\ell))$.
  \item Sparsity-preserving sampling:
  We are given that each node $v$ has 
  at least $\ell d(v)$ non-edges within its neighborhood) and a parameter $d'$. 
We seek a subset $S \subset V$ such that each node has: a) at most $d'$ neighbors in $S$ and b) at least $\ell' d'$ edges within $G[S_v]$, where $\ell' = \Omega(\ell)$ is maximized. \Risk: $\exp(-\Omega(\ell))$.
  \item Density splitting:  Given that each node has at least $\ell d(v)$ edges within its neighborhood,  partition the vertices into two sets $S_1$, $S_2$, such that each node has at least $\ell d(v)/16$ non-edges within $N(v) \cap S_1$ and within $N(v) \cap S_2$. \Risk: $\exp(-\Omega(\ell))$.

  \item Splitting a set with a large matching: Given that each node has a matching of size $\ell$ in its neighborhood, 
  partition the vertices into two sets $S_1$, $S_2$, such that each node has a matching of size $\ell/64$ within $G[N(v) \cap S_1]$ and within $G[N(v) \cap S_2]$. 
  \end{enumerate}
  \label{L:problem-bounds}
\end{lemma}

\begin{proof}
\begin{enumerate}
\item  Combine \cref{L:LLL-sample-bounds} parts (1) and (2), using \Cref{L:retraction-cost-union} and  \cref{obs:subsampling}. 
\item  We use sampling with fair coins for each node. So, in the context of a node $v$, let $X_1$, $X_2$ be the two parts of $N[v]$ formed by the sampling. The expected sparsity within either set $X_i$ is $\ell/4$. 
The bad event $\cE$ is when the sparsity within either part is less than $\ell/16$ and
the associated event $\cE'_v$ is the stricter event that the sparsity within either part is less than $\ell/6$. 
The probability of $\cE'_v$ is $\exp(-\Omega(\ell))$, by \cref{lem:nonEdgeHittingV2}.
If $\cE'_v$ occurs, then all incident variables are retracted, in which case the marginal probability is at most that of $\cE$ itself, or $\exp(-\Omega(\ell))$.
So, assume $\cE'_v$ did not occur. Now allow that an arbitrary subset $T$ of $N(v)$ gets retracted (without using the property of the second round of retractions). 
Let $X'_1 = X_1 \setminus T$ and $X'_2 = X_2 \setminus T$.
Since $\cE'_v$ did not occur, it holds for either $i$ that the set $X'_i \cup T$ has sparsity at least $\ell/6$. Now each node in $T$ is flipped again with a fair coin, producing sets $T_1$ and $T_2$ (that are r.v.'s). Let $S_i = X'_i \cup T_i$, for $i=1,2$.
The expected size of $S_i$ is at least $\ell/24$ (since only a subset of $T$ gets thrown out of $X'_i \cup T$ in the formation of $S_i$, each node with probability 1/2). Thus, the probability that $S_i$  has sparsity less than $\ell/32$ is $\exp(-\Omega(\ell))$, by \cref{lem:nonEdgeHittingV2}. Since this holds for every possible retraction, the risk is bounded above by $\exp(-\Omega(\ell))$.

\item 
Follows by \Cref{L:LLL-sample-bounds} part (1) and (3), combined via \Cref{L:retraction-cost-union}.

\item 
Identical to case 2, switching the roles of edges and non-edges.

\item  Fix a particular matching $M_v$ within $G[N(v)]$ of size $\ell$. We use again sampling with fair coins for each node, resulting in an initial partition of $N(v)$ into vertex sets $X_1$ and $X_2$. Let $Y_i$ be the number of edges of $M_v$ with both endpoints in $X_i$, for $i=1,2$.
Both endpoints of an edge in $M_v$ are in $Y_i$ with probability $1/4$, so
$\Exp[Y_i] = \mu = |M_v|/4 = \ell/4$. 
The bad events are $\cE = \cE^1 \cup \cE^2$ with $\cE^i$ denoting that $Y_i < \mu/8$.
Consider the stricter associated events $\cE'$ that either set satisfies $Y_i < \mu/2$. 
By Chernoff, $\Pr[\cE'] \le \Pr[\cE] = exp(-\Omega(|M_v|))$. 
Our argument now follows part (2) closely.
Consider the case when a node was happy with the assignment $X_1, X_2$, but afterward, a (possibly empty) subset $T$ of incident variables was retracted. 
To bound the risk, we want to bound the conditional probability of $\cE$ given the assignment fixed on $N(v) \setminus T$.
The subgraph induced by $X_i \cup T$ had a matching of size $\mu/2$.
When flipping the subset $T$ to produce sets $X'_1, X'_2$, in expectation at least a $\mu/8$-sized matching remains in $X'_i$. Thus, by Chernoff, the probability 
$X'_i$ contains no matching of size at least $\mu/16$ is $\exp(-\Omega(\mu))$.
Hence, the \risk is $\exp(-\Omega(\ell)$.
\end{enumerate}
\end{proof}
The proofs of the splitting problems (2), (4), and (5) in \Cref{L:problem-bounds} do not explicitly use the monotonicity of some of the events. Instead, they rely on a hereditary property of the events.

\section{Applications II: Coloring Sparse Graphs and Slack Generation}
\label{sec:coloringSparseAndSlackGeneration}
\subsection{Degree+1 List Coloring (d1LC)}
In the $deg+1$-list coloring (d1LC) problem, each node of a graph receives as input a list of available colors whose size exceeds its degree. The goal is to compute a proper vertex coloring in which each node outputs a color from its list. In the centralized setting, the problem can be solved with a simple greedy algorithm and it also admits efficient distributed algorithms. 
\begin{lemma}[List coloring \cite{HKNT22,HNT22}]
	\label{lem:listColoring}
	There is a randomized \CONGEST algorithm to $(deg+1)$-list-color (d1LC) any graph in $O(\log^5 \log n)$ rounds, w.h.p.  This reduces to $O(\log^3 \log n)$ rounds when the degrees and the size of the color space is $\poly(\log n)$.
\end{lemma}

In this section, we present our algorithms for degree-bounded sparsity-preserving sampling, slack generation, and our results on coloring triangle-free and sparse graphs with $\ll\Delta$ colors. 
\subsection{Sparsity-preserving Degree Reduction}
\label{sec:coloringSparseGraphs}

\myparagraph{Local Sparsity.} The \emph{(local) sparsity} $\zeta_v$ of a vertex $v$ is defined as $\zeta_v=\frac{1}{\Delta}\left(\binom{\Delta}{2}-m(N(v))\right)$, where for a set $X$ of vertices, $m(X)$ denotes the number of edges in the subgraph $G[X]$. Roughly speaking, $\zeta_v$ is proportional to the number of missing edges $\overline{m}_v$ in the graph induced by $v$'s neighborhood. For a node with degree $\Delta$, we have exactly $\overline{m}_v = \Delta \zeta_v$ non-edges in $G[N(v)]$. In general, the number of non-edges is $\binom{d(v)}{2} - m(N(v))$, where $m(N(v)) \le \binom{\Delta}{2} - \zeta_v \Delta$ for nodes with sparsity at least $\zeta_v$.

We require the following theorem to analyze how sparsity is preserved when randomly sampling nodes into a set. 

\begin{theorem}[Janson's inequality] 
    \label{thm:janson}
    Let $S\subseteq V$ be a random subset formed by sampling each $v \in V$ independently with probability $p$. Let $\mathcal{A}$ be any collection of subsets of $V$. For each $A \in \mathcal{A}$, let $I_A := \One(A \subseteq S)$ be an indicator variable for the event that $A$ is contained in $S$. Let $f := \sum_{A \in \mathcal{A}} I_A$ and $\mu := \E[f]$. Define
    \begin{equation}
        \label{eq:jansonK}
        K := \frac{1}{2}\sum_{A, B \in \mathcal{A}, A \cap B \neq \emptyset} \E[I_A I_B]
    \end{equation}
    Then, for any $0 \le t \le \E[f]$, 
    $$
        \Pr[f \le \E[f] - t] \le \exp \left(-\frac{t^2}{2 \mu + K}\right)
    $$
\end{theorem}

The expected number of non-edges that is preserved when sampling nodes into a set $S$ with probability $p$ is a $p^2$-fraction. The following lemma shows that the probability of deviating from this expectation is small.
\begin{restatable}[Non-edge hitting Lemma]{lemma}{lemNonEdgeHitting}
	Let $G$ be a graph on the vertex set $X$ with $\overline{m}$ non-edges. Sample each node of $X$ with probability $p$ into a set $S$ and let $f$ be the random variable describing the number of non-edges in $G[S]$. 
	Then we have $\Pr(f \le p^2 \overline{m}/2) \le \exp \left(-p \overline{m} / 5 |X| \right)$.
\label{lem:nonEdgeHittingV2}
\end{restatable}
\begin{proof}
    Let $\overline{E}$ be the set of non-edges in $G$. For each non-edge $e \in \overline{E}$, define an indicator variable $I_e = \One(e \subseteq S)$ for the event that the non-edge $e$ is preserved in $G[S]$. We have $\E[f] = \sum_{e \in \overline{E}} \E[I_e] = \overline{m} p^2$. We can bound (\ref{eq:jansonK}):  
    \begin{align*}
        K = \frac{1}{2}\sum_{e,e' \in \overline{E}, e \cap e' \neq \emptyset} \E[I_e I_{e'}] \le \frac{1}{2} \overline{m} (2 |X| - 2) p^3 \le \overline{m} |X| p^3
    \end{align*}
    Using \Cref{thm:janson} with $t=\E[f]/2$, we can compute $\Pr[f \le p^2\overline{m} / 2] = \Pr[f \le \E[f] - t]$, where 
    $\Pr[f \le \E[f] - t] 
        \le \exp \left(-\frac{t^2}{2 \E[f] + K} \right) 
        \le \exp \left(-\frac{p^4 \overline{m}^2 / 4}{2p^2 \overline{m} + p^3\overline{m} |X| } \right) 
        = \exp \left(-\frac{p^2 \overline{m}}{8 + 4p |X| } \right) 
        \le \exp \left(-\frac{p \overline{m}}{5 |X| } \right) 
    $
    assuming $p|X|\ge 8$.
  \end{proof}

We use the following lemma to compute a small-degree subgraph that preserves a large number of non-edges: 


\begin{lemma}[Degree-Bounded Sparsity-Preserving Sampling]
\label{lem:sparsityPreserving}
    Assume that $\log^2\log n \le \Delta \le O(\log n)$. Let $X, Y \subseteq V$ such that for all $v \in X$, the number of non-edges in $G[N(v) \cap Y]$ is at least $\alpha\Delta^2$ for some $0 < \alpha \le 1/2$ with $1/\alpha = O(\poly\log\log n)$. 
    Let $\mu= (600/\alpha)\log\Delta \cdot \log\log n$. 
    There is a randomized $\poly \log \log n$-time \CONGEST algorithm, that w.h.p. finds a set $S \subseteq Y$ s.t. each $v \in X$ has at most $\Delta_s = 4\mu$ neighbors in $S$, and at least $\overline{m}_{thres}=\alpha \mu^2 / 2=\Omega((1/\alpha) \log^2\Delta \cdot \log^2 \log n)$ non-edges in subgraph induced by $N(v) \cap S$.
\end{lemma}
\begin{proof}
    We define a \textit{sparsity-preserving degree reduction LLL} and solve it with the algorithm of \Cref{thm:promiseLLL} for binary LLLs with low risk. 
    For each $w \in Y$, define a variable indicating that $w$ is sampled to $S$, which happens with probability $p:=\mu/\Delta$. 
    Sampled nodes are called \black and non-sampled nodes \white. 
    For each $v \in X$, define two unwanted events $\cE_d$ and $\cE_\zeta$:
    \begin{itemize}
        \item Let $\cE_d$ be the event that $v$ has more than $4\mu$ sampled neighbors in $S$.
        The expected number of neighbors in $S$ is $d(v) \cdot \mu / \Delta \le \mu$. 
        Hence, $\Pr(\cE_d) \le \exp\left(-2\mu/3\right)$ by Chernoff.
        Additionally, define an associated event $\assoc(\cE_d)$ as the event that at most $2\mu$ neighbors are sampled. We have $\Pr(\assoc(\cE_d)) \le \exp\left(-\mu/3\right)$. This bounds the \risk of $\cE_d$ to be at most $\Pr(\assoc(\cE_d))$ by  \Cref{L:incr-LLL} 
        \item Let $\cE_\zeta$ be the event that the number of non-edges in $G[N(v) \cap S]$ is less than $\overline{m}_{thres} = \alpha \mu^2 / 2$. Let $f$ be a random variable for the number of non-edges in the graph induced by $X = N(v) \cap S$. Apply \cref{lem:nonEdgeHittingV2}, with $|X| \le \Delta $ and $\overline{\mu} \ge \alpha\Delta^2$. We have $\E[f] \ge p^2 \overline{m} \ge \alpha \mu^2$. This gives 
        $
        \Pr(\cE_\zeta) = \Pr(f \le \E[f] / 2) 
        \le \exp \left( -\frac{p \overline{m}}{5 |X|} \right) 
        \le \exp \left( -\frac{\alpha \mu}{5} \right) 
        $. $\cE_\zeta$ is a monotone increasing event.
        Hence, its \risk is at most $p_{\zeta}=\Pr(\cE_\zeta)$ by \Cref{L:monotone-incr}, where the associated event $\assoc(\cE_\zeta)$ is $\cE_\zeta$ itself. 
    \end{itemize}
    For each $v \in X$, a bad event $\cE$ is defined as the union $\cE_d \cup \cE_\zeta$. This event depends on the sampling status of the neighbors of $v$. 
    Hence, the dependency degree of the sparsity-preserving sampling LLL is $d=\Delta^2$.
    By \Cref{L:retraction-cost-union}, the associated event of $\cE$ is the union, $\mathsf{assoc}(\cE)=\mathsf{assoc}(\cE_d) \cup \mathsf{assoc}(\cE_\zeta)$. 
    The \risk of $\cE$ is at most $p_{\zeta} + p_d = \Pr(\cE_\zeta) + 2\Pr(\cE_d) \le e^{-\mu/3} + 2e^{-\alpha\mu/5} \le d^{100\cdot \log\log n}$. 

    The locality of the LLL is 2, since dependent events are within two hops of each other in $G$. We show that the LLL is simulatable. We can show that the required primitives in \Cref{def:simulatability} can be executed in $O(\mu)=\poly\log\log n$ rounds: 
    \begin{enumerate}
        \item \textit{Test}: Sampled nodes inform their neighbors. If an event node $v \in X$ has more than $2\mu$ sampled neighbors, the event $\assoc(\cE)$ occurs. Otherwise, $v$ pipelines a list of IDs of its sampled neighbors $L=N(v) \cap S$ to each of its (sampled) neighbors in $O(\mu)$ rounds. Each neighbor $w \in N(v) \cap S$ receiving the list report which nodes $x \in L$ are not in $N(w)$. This allows $v$ to learn the list of non-edges in $G[N(v) \cap S]$.
        \item \textit{$1$-bit Min-Aggregation}: Trivial to do in one round, as variables are adjacent to events in $G$. 
    \end{enumerate}
    For the following, we can assume that nodes have access to $\mathsf{IDbitLen}=\Theta(\poly\log\log n)$-bit IDs. 
    \begin{enumerate}
        \item [3] \textit{Evaluate}: We start with some pre-processing for each $v \in X$ to learn the edges in $G[N(v)]$ (all parallel instances share a common pre-processing phase). All nodes $v \in V$ forward the list of IDs in $N(v)$ to each neighbor. The IDs take at most $O(\mathsf{IDbitLen} \cdot \Delta)=O(\log n \cdot \poly\log\log n$ bits, which can be sent in $\poly\log\log n$ rounds. Given partial assignments $\phi$ and $\psi$, any event node $v \in X$ can compute the required conditional probabilities locally, using knowledge of the structure of $G[N(v)]$. 
        \item [4] \textit{Min-aggregation}: This is trivial as variables are adjacent to events in $G$. Sending the $O(\log n)$ different $O(\log\log n)$-bit strings takes $O(\log\log n)$ rounds. 
    \end{enumerate}
\end{proof}

\subsection{Slack Generation for Sparse Nodes}

The slack of a node (potentially in a subgraph) is defined as the difference between the size of its palette and the number of uncolored neighbors (in the subgraph). 

\begin{definition}[Slack]
    \label{def:slack}
    Let $v$ be a node with color palette $\Psi(v)$ in a subgraph $H$ of $G$. The slack of $v$  in $H$ is the difference $|\Psi(v) - d|$, where $d$ is the number of uncolored neighbors of $v$ in $H$. 
\end{definition}

Slack generation is based on trying a random color for a subset of nodes. Sample a set of nodes and a random color for each of the sampled nodes. Nodes keep the random color if none of their neighbors chose the same color. See \Cref{alg:slackgen} for a pseudocode. 

Consider a node $v$ with lots of sparsity in its neighborhood, i.e., many non-edges in $G[N(v)]$. Let $w, w' \in N(v)$ be two neighbors of $v$ such that $w,w'$ are not adjacent. If $w$ and $w'$ keep the same color, $v$ gets one unit of slack, as two of its neighbors get colored while the color palette is only reduced by one color. It can be shown that the obtained slack is roughly proportional to the sparsity of $v$, see for example \cite{EPS15, HKMT21}. We prove a similar result in terms of the number of non-edges, with a custom-sized color palette: 

\begin{algorithm}[t]\caption{{\trycolor} (vertex $v$, color $c_v$)}\label{alg:trycolor}
\begin{algorithmic}[1]
\STATE Send $c_v$ to $N(v)$, receive the set $T=\{c_u : u\in N(v)\}$.
\STATE{\textbf{if}} $c_v\notin T$ \textbf{then} permanently color $v$ with $c_v$.
\STATE Send/receive permanent colors, and remove the received ones from $\Psi(v)$.
\end{algorithmic}
\end{algorithm}

\begin{restatable}[Slack generation with custom color space]{lemma}{lemSlackGeneration}
    \label{lem:slackgen-custom}
    Let $\Delta_s, \chi$ be positive integers with $\chi \ge c'\Delta_s$ for some constant $c' > 0$. Let $S \subseteq V$ be a subset of nodes.  
    Consider a node $v \in V$ with at least $\overline{m}$ non-edges in $G[N(v) \cap S]$. Suppose that all nodes in $S$, as well as $v$, have at most $\Delta_s$ neighbors in $S$. After running SlackGeneration on $S$ with color palette $[\chi]$, the slack of $v$ is increased by at least  $e^{-3/c'} \overline{m} /(500 \chi) = \Omega(\overline{m}/\chi)$, with probability at least $1-\exp(-\Omega(\overline{m}/\chi))$. This holds independently of the random choices at a distance greater than 2. 
\end{restatable}
\begin{proof}
    Let $N_S(v) = N(v) \cap S$ for short. 
    Let $X \subseteq \binom{N_S(v)}{2}$ be the set of non-edges, where $|X| = \overline{m}$. 
    Each node in $S$ is activated with probability $p=1/20$. 
    Activated node $w$ runs \trycolor, where $w$ selects a color u.a.r. from $[\chi]$. 
    
    Let $c_w$ be the random color chosen, and let $x_w$ be the possible permanent color, or $x_w=\bot$ in case of a conflict. 
    
    Let $Z$ be the number of colors $c \in [\chi]$ s.t. there exists a non-edge $\{u,w\} \in X$ with $c_u=c_w=c$, and \emph{for all} such non-edges, $x_u=x_w=c$, i.e. all the nodes retain the color (see below for a formal definition with quantifiers). Say that a non-edge $\{u,w\} \in X$ is \textit{successful} if $x_u=x_w \neq \bot$, and no node in $N_S(v) \setminus \{u,w\}$ picks the same color.
    Let $Y_{\{u,w\}}$ be an indicator function for the event that $\{u,w\}$ is successful. We have $\E[Z] \ge \sum_{\{u,w\} \in X} \E[Y_{\{u,w\}}]$, since each non-edge with $Y_{\{u,w\}}=1$ counts towards $Z$. The probability of a non-edge being successful is at least 
    $$ 
        \Pr(Y_{\{u,w\}}) 
        \ge p^2 \left(\frac{\chi - 1}{\chi}\right)^{2\Delta_s - 2} \left(\frac{1}{\chi}\right) \left(\frac{\chi - 1}{\chi}\right)^{\Delta_s - 2} 
        \ge \frac{p^2}{\chi} \left(\frac{\chi - 1}{\chi}\right)^{3\Delta_s} 
        \ge \frac{e^{-3/c'} p^2}{\chi} 
    $$  
    where $u$ and $v$ are activated with probability $p^2$ and select the same color w.p. $1/\chi$; with probability $(\frac{\chi - 1}{\chi})^{2\Delta_s - 2}$ none of the neighbors of $u$ nor $v$ choose that particular color, and $(\frac{\chi - 1}{\chi})^{\Delta_s - 2}$ is the probability that no other neighbors of $v$ choose that color. In the last inequality, we used that $\chi \ge c'\Delta_s$. This gives
    $ \E[Z] \ge \sum_{\{u,w\} \in X} \E Y_{\{u,w\}} \ge \overline{m} e^{-3/c'} p^2 / \chi$. 

    Next, we show that $Z$ is concentrated around its mean. 
    Let $T$ be the number of colors that are randomly chosen in \trycolor by both nodes of at least one non-edge in $X$. Let $D$ be the number of colors that are chosen in \trycolor by both nodes of at least one non-edge in $X$, but are not retained by at least one of them. Formally, 
    \begin{alignat*}{2}
        T&:= \# \text{ colors } c \text{ s.t. } &&\exists \{u,w\} \in X: c_u=c_w=c \\
        D &:= \# \text{ colors } c \text{ s.t. } &&\big(\exists \{u,w\} \in X: c_u=c_w=c\big) \;\wedge \\
        & &&\big(\exists \{u,w\} \in X: (c_u=c_w=c) \wedge (x_u = \bot  \vee  x_w = \bot)\big) \\
        Z&:= \# \text{ colors } c \text{ s.t. } &&\big(\exists \{u,w\} \in X: c_u=c_w=c \big) \;\wedge \\
        & &&\big(\forall \{u,w\} \in X: (c_u=c_w=c) \implies (x_u=x_w=c)\big)
    \end{alignat*}
    We have $Z=T-D$, since $D$ counts the colors where the implication in the definition of $Z$ fails. 

    We upper bound $\E[T]$ (which implies the same bound for $D$, as $D \le T$). For a fixed color $c$, the probability that both $u,w$ pick $c$ is at most $1/\chi^2$. By union bound, the probability that $c$ is picked by at least one non-edge is at most $\overline{m} / \chi^2$. There are $\chi$ colors, so $\E[T] \le \chi \cdot \overline{m} / \chi^2 = \overline{m} / \chi$. 

    The functions $T$ and $D$ are $r$-certifiable with $r=2$ and $r=3$, respectively. See the appendix and \Cref{lem:talagrand} for the definition of an $r$-certifiable function. $T$ and $D$ are both $2$-Lipschitz: whether a node is activated, and which color it picks affects the outcome by at most $c=2$. 
    We apply \Cref{lem:talagrand} with $b = \E[Z] / 10 - 60c\sqrt{r \cdot \E[T]}$:
    \begin{align*}
        \Pr[|T-\E[T]| \ge b + 60 c \sqrt{r \cdot \E[T]}] 
        &= \Pr[|T-\E[T]| \ge \E[Z] / 10] \\
        &\le 4\exp\left(-\frac{\big(\E[Z] / 10 - 60c\sqrt{r \cdot \E[T]}\big)^2}{8c^2 r \E[T]}\right) \\
        &\le \exp\left(-\Theta(1) \left(\frac{\E[Z]^2}{\E[T]} - \frac{\E[Z]}{\sqrt{\E [T]}} + O(1)\right)\right) \\
        &\le \exp\left(-\Omega(\overline{m}/\chi)\right)
    \end{align*}
    In the last inequality, we used that $\E[Z] \ge \overline{m} e^{-3/c'} p^2 / \chi$ and $\E[T] \le \overline{m} / \chi$. The same concentration bound applies for $D$, meaning that $\Pr[|D-\E[D]| \ge \E[Z] / 10] \le \exp\left(-\Omega(\overline{m}/\chi)\right)$. By union bound, neither of the events $|T - \E[T]| \ge \E[Z]/10$ and $|D - \E[D]| \ge \E[Z]/10$ occur with probability at least $1 - 2\exp\left(-\Omega(\overline{m}/\chi)\right)$. 
    Hence, $Z = T-D \ge \E [T] - \E[Z]/10 - (\E [D] + \E [Z] / 10) = (4/5) \cdot \E [Z] \ge e^{-3/c'} \overline{m} /(500 \chi)$, with probability at least $1-\exp\left(-\Omega(\overline{m}/\chi)\right)$.
\end{proof}

In the next lemma, we use our disjoint variable set LLL to produce large amounts of slack for nodes. The lemma assumes that we are already given two sets $S_1$ and $S_2$ that induce sufficiently many non-edges in the neighborhood of every relevant node. 
\begin{lemma}
    \label{lem:slackgenAlg}
    Let $\Delta_s = O(\poly\log\log n)$. Let $\overline{m}$ and $\chi=O(\Delta)$ be positive integers such that $\overline{m} / \chi = \Omega(\log \Delta \cdot \log\log n)$ and $\chi \ge c'\Delta_s$ for some constant $c'$. Let $W \subseteq V$ and  let $S_1, S_2 \subset V$ be disjoint sets such that for $i=1,2$, 
    \begin{itemize}
        \item $\forall v \in (W \cup S_1 \cup S_2): d_{S_i}(v) \le \Delta_s$,
        \item $\forall v \in W$: the number of non-edges in $N(v) \cap S_i$ is at least $\overline{m}$
    \end{itemize}
    There is a randomized \CONGEST algorithm that w.h.p. colors a subset of $S_1 \cup S_2$ using a palette of size $2\chi$  such that every node in $W$ has at least $e^{-3/c'} \overline{m} /(500 \chi) = \Omega(\overline{m}/\chi)$ same-colored neighbors. Every node in $W, S_1, S_2$ has at most $2\Delta_s$ of its neighbors colored. 
\end{lemma}
\begin{proof}
    We define a disjoint variable LLL for slack generation and solve it with \Cref{thm:LLLTwoSets}. 
    The random process is defined by the \slackgeneration algorithm. 
    We define two variables for each $v \in S_1 \cup S_2$. 
    Let $a_v$ be a variable indicating that $v$ is activated, which happens with probability $1/20$. 
    Let $c_v$ be a variable for a color chosen uniformly at random from the node's palette, ${1, \dots, \chi}$ and ${\chi+1, \dots, 2\chi}$ for nodes in $S_1$ and $S_2$, respectively. 
    For notational purposes, $c_v$ is defined for all nodes, regardless of the value of $a_v$. 
    The variables $a_v, c_v$ are also independent of the variables of other nodes.
    For each $v \in W$, define the bad event $\cE_v$ that the slack of $v$ does not get increased by at least $s_{thres} := e^{-3/c'} \overline{m} /(500 \chi)$. 
    The probability of $\cE_v$ is at most $\exp(-\Omega(\overline{m}/\chi))$ by \Cref{lem:slackgen-custom}.

    We show that the slack generation LLL is simulatable. 
    \begin{enumerate}
        \item \textit{Test}: Each variable node $v \in S_1$ runs \slackgeneration. Activated nodes that retained their color inform their neighbors. An event node $w \in W$ receives the retained colors in its neighborhood and counts how much slack was generated. 
        \item \textit{1-bit min-aggregation}: Each event broadcasts its bit, traveling for 2 hops. Broadcasts of multiple events can be combined (since any variable node $v \in S$ is a variable for all events in its 2-hop neighborhood). The same works in the other direction. 
    \end{enumerate}
    By the definition of simulatability, we can assume that each event and variable has access to a unique $O(\log\log n)$-bit identifier $\mathsf{id}$ for the following: 
    \begin{enumerate}
        \item [3.] \textit{Evaluate}: We show the primitive for a single partial assignment using $\poly\log\log n$ bandwidth -- the primitive for $O(\log n)$ parallel instances follows from this. Let $S=S_1 \cup S_2$. Each event node $w \in W$ can learn the list of neighbors in $S$ for each of its immediate neighbors $v \in N(w) \cap S$ in $O(\Delta_s)=\poly\log\log n$ rounds (this will even be the same for all parallel instances). Given a partial assignment $\psi$, where each variable node $v \in S$ knows its values (or $\bot$) for $a_v$ and $c_v$, $w$ learns the values in $N^2(w) \cap S$: each immediate neighbor $v \in N(w)$ pipelines $L=\{(\psi(x), \mathsf{id}(x)) : x \in N(v) \cap S\}$ to $w$. Communicating $L$ takes $O(\Delta_s \cdot \chi)=\poly\log\log n$ bits. Knowing the structure of the distance-2 neighborhood in $S$, as well as the partial assignment for nodes in $N^2(w) \cap S$, allows $w$ to locally compute the conditional probability of obtaining enough slack, i.e. its event not occurring. 
        \item [4] \textit{Min-aggregation}: The principle is the same as for 1-bit aggregation. Sending the $O(\log n)$ different $O(\log\log n)$-bit strings takes $O(\log \log n)$ rounds. 
    \end{enumerate}
\end{proof}

\subsection{Coloring Sparse Graphs}
In this section, we prove the following theorem.
\begin{theorem}[Coloring sparse graphs]
    \label{thm:sparseColoring}
    Let $G$ be a graph with maximum degree $\Delta = \Omega(\log^2 \log n)$ and sparsity $\zeta \ge \epsilon^2 \Delta$. Let $x=\log\Delta \cdot \log\log n / 6$. There is a randomized $\poly\log\log n$-time \CONGEST algorithm computing a $(\Delta-x)$-coloring of $G$ with high probability. 
\end{theorem}
\begin{proof}
    Let $\Delta'=\Delta - x$. Let $D=\{v\in V: d(v) \ge \Delta'\}$. As a node $v \in D$ with $d(v) \ge \Delta'$ initially has more neighbors than available colors, we need to increase its slack by at least $d(v) + 1 - \Delta' \le \Delta+1-\Delta'=x+1$ to get a $(\deg + 1)$-list coloring instance. The sparsity of any node $v$ is $\zeta = \frac{1}{\Delta}\big(\binom{\Delta}{2} - m(v)\big) \ge \epsilon^2 \Delta$, which implies that the number of edges $m(v)$ in $G[N(v)]$ is at most $\binom{\Delta}{2} - \Delta^2 \epsilon^2$. For a node $v \in D$, the number of non-edges in $G[N(v)]$ is at least $\binom{\Delta'}{2} - m(v) \ge \binom{\Delta'}{2} - \binom{\Delta}{2} + \epsilon^2\Delta^2 \ge \frac{1}{2}\epsilon^2 \Delta^2 \eqqcolon \overline{m}$
    where the last inequality assumes $x \le \epsilon^2 \Delta / 2$, which holds for large enough $n$.  

    Assume that $\Delta=\Omega(\log^2\log n)$ and $\Delta \le 100\log n$. Start by computing two disjoint sets $S_1, S_2 \subseteq V$, that preserve the sparsity for nodes in $D$, while having a bounded degree. Apply \Cref{lem:sparsityPreserving} for $X=D$ and $Y=V$, where the number of non-edges is at least $\overline{m}=\epsilon^2\Delta^2/2$. The result is a set $S_1 \subseteq Y$ s.t. all nodes have at most $(4800/\epsilon^2) \log\Delta \cdot \log\log n$ neighbors in $S_1$ and each $v \in D$ has at least $(10^5 / \epsilon^2) \log^2 \Delta \cdot \log^2 \log n$ non-edges in $G[N(v) \cap S_1]$. The algorithm works with high probability and runs in $\poly\log\log n$ rounds. 
    The number of non-edges in the remaining graph $G[N(v)\cap (V \setminus S_1)]$ for $v \in D$ is at least $\epsilon^2 \Delta^2 / 2 - \Delta \cdot \Delta_s \ge \epsilon^2 \Delta^2 / 4$ when $n$ is large enough. Compute another sparsity preserving set $S_2 \subset V \setminus S_1$ with the same approach. All nodes have at most $\Delta_s := (9600/\epsilon^2) \log\Delta \cdot \log\log n$ neighbors in $S_2$ and each $v \in D$ has at least $\overline{m}_s := (10^5 / \epsilon^2) \log^2 \Delta \cdot \log^2 \log n$ non-edges in $G[N(v) \cap S_2]$.

    We generate slack for nodes in $D$ by running \Cref{lem:slackgenAlg} with the sets $S_1, S_2$. We use a color palette of size $2\Delta_s$ (divided equally for the two sets). This colors a subset of $S_1 \cup S_2$ (possibly including nodes in $D$), such that the slack for each $v \in D$ is increased by at least $\overline{m}_s / (1000 e^3 \cdot \Delta_s) \ge (1/5) \log\Delta \cdot \log\log n$. This is at least the required slack, $x+1$. At this point, we can color the remaining uncolored nodes in $D$ and $V \setminus D$ (at the same time) with the (deg+1)-list coloring algorithm of \Cref{lem:listColoring}.

    Lastly, suppose that $\Delta \ge 100\log n$. Run \slackgeneration on $G$, using all $\Delta'$ colors. It was previously shown that each $v \in D$ has at least $\overline{m} = \frac{1}{2}\epsilon^2 \Delta^2$ non-edges in its neighborhood. By \Cref{lem:slackgen-custom}, each $v \in D$ has its slack increased by at least $
    \frac{1}{1000 e^6} \frac{\epsilon^2 \Delta^2}{\Delta - x} \ge \frac{1}{10^6} \epsilon^2 \Delta$, with probability at least $1-\exp(-\Omega(\overline{m}/\Delta')) \ge 1 - \exp(-\Omega(\Delta)) \ge 1-n^{-c}$ for some constant $c$. The obtained slack $\frac{1}{10^6} \epsilon^2 \Delta$ is at least the required $\log\Delta \cdot \log\log n +1$ when $n$ is large enough. 
\end{proof}

\subsection{Coloring Triangle-free Graphs}
Triangle-free graphs are maximally sparse in the sense that neighbors of a node are never connected. This allows us to generate slack linear in $\Delta$. We prove the following theorem.
\begin{restatable}[Coloring Triangle-free Graphs]{theorem}{thmTriangleFreeColoring}
    \label{thm:triangleFreeColoring}
    Let $G$ be a triangle-free graph with maximum degree $\Delta$. There is a randomized $\poly\log\log n$-time \CONGEST algorithm to $\gamma\Delta$-color $G$ with high probability, for a constant $\gamma=1-10^{-7}$.
\end{restatable}
\begin{proof}
    Let $\Delta'=\gamma\Delta$. Let $D=\{v\in V: d(v) \ge \Delta'\}$. We need to generate slack for all nodes in $D$. As a node $v \in D$ with $d(v) \ge \Delta'$ initially has more neighbors than available colors, we need to increase its slack by at least $d(v) + 1 - \Delta' \le \Delta+1-\Delta'$ to get a $(\deg + 1)$-list coloring instance. 

    We consider three different ranges of $\Delta$: (1) $\Delta=O(\log^2\log n)$, (2) $\Delta=\Omega(\log^2\log n)$ and $\Delta \le 100\log n$, and (3) $\Delta \ge 100\log n$. We start by giving a proof for (2). In the other cases, the result follows from previous work, or all events hold with high probability. 

    Assume that $\Delta=\Omega(\log^2\log n)$ and $\Delta\le100\log n$.
    Split all vertices into $k=2\epsilon^4 \Delta / (\ln \Delta \log^2 \log n)$ classes with discrepancy $\epsilon \Delta/ k$ using the algorithm of \cite[Theorem 23]{HMN22}, for some small constant $\epsilon$. The algorithm runs in $\poly\log\log n$ rounds, with high probability. It produces a partition $V_1, \dots , V_k$ of $V$ s.t. for all $v \in V$ and all $1 \le i \le k: d_{V_i}(v) = d(v) / k \pm \epsilon \Delta / k$. In particular, for all $v \in D$, the number of neighbors in $V_i$ is at most $\Delta_s := \frac{\Delta}{k} + \frac{\epsilon \Delta}{k} = O(\ln \Delta \cdot \log^2 \log n)$ and at least $\delta_s := \frac{\gamma\Delta}{k} - \frac{\epsilon \Delta}{k} = \Omega(\ln \Delta \cdot \log^2 \log n)$. As there are no triangles, for each $v \in D$ the number of non-edges in  $G[N(v) \cap V_i]$ is at least $\overline{m} := \binom{\delta_s}{2} = \frac{1}{2}\big(\frac{\gamma \Delta}{k} - \frac{\epsilon \Delta}{k}\big)\big(\frac{\gamma \Delta}{k} - \frac{\epsilon \Delta}{k}-1\big) \ge \frac{\Delta^2}{16k^2}$ assuming $\gamma - \epsilon \ge 1/2$. 

    We generate slack on the classes of the partition in parallel. 
    For $1 \le i \le k$, each class $V_i$ is assigned a subset of $\chi=\lfloor\frac{\Delta'}{k}\rfloor$ colors, $[ (i-1) \lfloor\frac{\Delta'}{k} \rfloor, i \cdot \lfloor\frac{\Delta'}{k}\rfloor]$. Form $k/2$ instances, each using a pair of vertex classes: for 
    $1 \le j \le k/2$, the nodes in $V_{2j-1}, V_{2j}$ are used together. 
    We apply \Cref{lem:slackgenAlg} on each instance in parallel. 
    The size of the color palette satisfies $\chi \ge c \Delta_s$ for $c=1/2$. 
    Each instance succeeds with high probability, and we take a union bound over all $k/2$ instances.
    To avoid congestion, we assign each instance a subset of the edges for communication, such that any edge is used for at most two instances. 
    Note that a node in $V_i$ only needs to communicate with other nodes in $V_i$ within 2 hops. Let $\{v,w\} \in E$ be any edge, where $v \in V_x$ and $w \in V_y$ for some $1\le x \le y \le k$. The edge $\{v,w\}$ is used for communication in the instances with $V_x$ and $V_y$, for a total of at most 2 instances. 
    
    By \Cref{lem:slackgenAlg}, in each instance, the amount of slack generated is at least 
    $\frac{e^{-3/c}\overline{m}}{500\chi} 
    \ge \frac{\Delta^2 / 16k^2}{500e^{6}\cdot \gamma\Delta/k} 
    \ge \frac{\Delta}{4\cdot 10^6\cdot k}$. 
    Over all the $k/2$ instances, the amount of slack generated for each $v \in D$ is at least $\Delta \cdot 10^{-7}$. 
    Hence, all nodes in $D$ get slack by at least the required amount, $d(v) + 1 - \Delta' \le \Delta + 1 - \gamma\Delta \approx 10^{-7} \cdot \Delta$.
    At this point, we can color the remaining uncolored nodes in $D$ and $V \setminus D$ (at the same time) with the $(\deg+1)$-list coloring algorithm of \Cref{lem:listColoring}.

    It remains to handle cases (1) and (3). When $\Delta \ge 100\log n$, slack is generated with high probability.
    Run \slackgeneration on $G$, using all $\chi=\Delta'$ colors.
    Each $v \in D$ has at least $\overline{m}=(\Delta')^2/2$ non-edges in its neighborhood.
    By \Cref{lem:slackgen-custom}, 
    each $v \in D$ has its slack increased by at least $\frac{\overline{m}}{500e^3\chi} \ge \frac{1}{21000}\Delta' \ge \Delta \cdot 10^{-7}$, 
    with probability at least $1-\exp(-\Omega(\overline{m}/\chi)) \ge 1-\exp(-\Omega(\gamma\Delta)) = 1-1/n^{c}$ for some constant $c$.
    Lastly, when $\Delta = O(\log^2\log n)$, this same process of slack generation is an exponential LLL with a small dependency degree, $d=\poly\log\log n$. This can be solved with high probability using the algorithm of \cite{MU21}, running in $\poly\log\log n$ rounds. 
\end{proof}

\bibliographystyle{alpha}
\bibliography{references}
\appendix

\clearpage
\section{Supplementary Results}

\subsection*{Concentration Bounds}

We use the following Talagrand's inequality. A function $f(x_1, \dots, x_n)$ is called $c$-\textit{Lipschitz} iff the value of any single $x_i$ affects $f$ by at most $c$. Additionally, $f$ is $r$-\textit{certifiable} if for every $x = (x_1, \dots, x_n)$, (1) there exists a set of indices $J(x) \subseteq [n]$ such that $|J(x)| \le r \cdot f(x)$, and (2) if $x'$ agrees with $x$ on the coordinates in $J(x)$, then $f(x') \ge f(x)$. 
\begin{lemma}[Talagrand's Inequality II \cite{molloy2013coloring}]
    \label{lem:talagrand}
    Let $X_1, \dots, X_n$ be independent random variables and $f(X_1, \dots, X_n)$ be a $c$-Lipschitz, $r$-certifiable function. For any $b \ge 1$:
    $$\Pr(|f-\E[f]| > b + 60c\sqrt{r \E[f]}) \le 4 \exp{\left(-\frac{b^2}{8c^2r\E[f]}\right)}$$
\end{lemma}

\subsection*{Deterministic LLL in LOCAL}
LLLs can be solved efficiently with deterministic algorithms in the \LOCAL model. 
\begin{theorem}[Deterministic LLL in \LOCAL, \cite{RG20}]
\label{thm:deterministicLLLLOCAL}
There is a deterministic \LOCAL algorithm for the constructive \lovasz local lemma under criterion $epd(1+\eps)<1$, for a constant $\eps>0$ that runs in $O(\log^* s)+\poly\log n$ rounds if the communication network has at most $n$ nodes and node IDs are from a space of size $s$ and the event/variable assignment has locality $\poly\log n$.
\end{theorem}
\Cref{thm:deterministicLLLLOCAL} is proven by using the powerful general derandomization framework of \cite{RG20,GHK18,SLOCAL17} for the algorithm of Moser-Tardos \cite{MoserTardos10}. 

\subsection*{Shattering Lemma}

\begin{lemma}[The Shattering Lemma, \cite{FGLLL17,BEPSv3}]\label{lem:Shattering}
Let $G=(V, E)$ be a graph with maximum degree $\Delta$. Consider a process that generates a random subset $B \subseteq V$ such that $P[v \in B]\leq \Delta^{-c_1}$, for some constant $c_1 \geq 1$, and such that the random variables $1(v\in B)$ depend only on the randomness of nodes within at most $c_2$ hops from $v$, for all $v\in V$, for some constant $c_2\geq 1$.
Then, for any constant $c_3\geq 1$, satisfying  $c_1>c_3+ 4c_2 + 2$,  we have that any connected component in $G[B]$ has size at most $O( \log_{\Delta} n  \Delta^{2c_2})$ with probability at least $1- n^{-c_3}$.
\end{lemma}

\subsection*{Communication on Top of Weak Network Decompositions}
\label{app:networkDecompositionRouting}
\begin{lemma}[\cite{GGR20,MU21}]
	\label{cor:treeAggregationBetter}
	Let $G$ be a communication graph on $n$ vertices. Suppose that each vertex of $G$ is
	part of some cluster $\cC$ such that each such cluster has a rooted Steiner tree $T_{\cC}$ of diameter at most
	$\beta$ and each node of $G$ is contained in at most $\kappa$ such trees. Then, in $O(\max\{1,\kappa/b \}\cdot (\beta + \kappa))$ rounds of the
	\CONGEST{}{} model with $b$-bit messages, we can perform the following operations for all
	clusters in parallel on all clusters:
	\begin{enumerate} 
		\item  Broadcast: The root of $T_{\cC}$ sends a $b$-bit message to all nodes in $\cC$;
		\item  Convergecast: We have $O(1)$ special nodes $u \in \cC$, where each special node starts with a
		separate $b$-bit message. At the end, the root of $T_{\cC}$ knows all messages;
		\item  Minimum: Each node $u \in \cC$ starts with a non negative $b$-bit number $x_u$. At the end, the root
		of $T_{\cC}$ knows the value of $\min_{u \in \cC}x_u$;
		\item  Summation: Each node $u \in \cC$ starts with a non negative $b$-bit number $x_u$. At the end, the root
		of $T_{\cC}$ knows the value of $\big(\sum_{u\in \cC} x_u\big) \mod 2 ^{O(b)}$.
	\end{enumerate} 
\end{lemma}

\end{document}